\documentclass[12pt]{article}
\usepackage{thm-restate}
\usepackage[margin=1.1in]{geometry}
\usepackage{graphicx}
\usepackage{array}
\usepackage{url}
\usepackage{times}
\usepackage{amssymb}
\usepackage{amsmath}
\usepackage{latexsym}
\usepackage{graphics}
\usepackage{url}
\usepackage{verbatim}
\usepackage{color}
\usepackage{lineno}
\usepackage{multicol}
\usepackage{caption}
\usepackage{subcaption}
\usepackage{lineno}
\usepackage{soul}
\usepackage{comment}
\usepackage{datenumber}
\usepackage{amsthm}
\usepackage{paralist}
\usepackage{enumerate}
\usepackage{enumitem}
\usepackage{thm-restate}
\usepackage[noend]{algpseudocode}
\usepackage{algorithm}

\usepackage{listings}
\usepackage{xcolor}

\definecolor{codegreen}{rgb}{0,0.6,0}
\definecolor{codegray}{rgb}{0.5,0.5,0.5}
\definecolor{codepurple}{rgb}{0.58,0,0.82}
\definecolor{backcolour}{rgb}{0.95,0.95,0.92}

\lstdefinestyle{mystyle}{
  backgroundcolor=\color{backcolour},   commentstyle=\color{codegreen},
  keywordstyle=\color{magenta},
  numberstyle=\tiny\color{codegray},
  stringstyle=\color{codepurple},
  basicstyle=\ttfamily\tiny,
  breakatwhitespace=false,         
  breaklines=true,                 
  captionpos=b,                    
  keepspaces=true,                 
  numbers=left,                    
  numbersep=5pt,                  
  showspaces=false,                
  showstringspaces=false,
  showtabs=false,                  
  tabsize=2
}

\usepackage{fancybox}
\usepackage{tikz}
\usepackage[T1]{fontenc}
\usepackage{fourier}

\usepackage[english]{babel}															
\usepackage[protrusion=true,expansion=true]{microtype}	
\usepackage{amsmath,amsfonts,amsthm} 

\usepackage{tikz}
\usetikzlibrary{tikzmark,decorations.pathreplacing,calc}
\usepackage{mathtools}

\DeclarePairedDelimiter\floor{\lfloor}{\rfloor}

\usepackage{fancyhdr}
\numberwithin{equation}{section}		
\numberwithin{figure}{section}			
\numberwithin{table}{section}				

\newtheorem{theorem}{Theorem}[section]          
\newtheorem{lemma}[theorem]{Lemma}             
\newtheorem{corollary}[theorem]{Corollary}         

\title{
\Large
Generating pivot Gray codes for spanning arborescences of complete digraphs and split digraphs in constant amortized time \\
}
\author{
Bowie~Liu  
\ \ \ \ 
Dennis~Wong 
}
\date{}

\begin{document}
\maketitle

\begin{abstract}
We present the first known efficient algorithm to generate pivot Gray codes for spanning arborescences of complete digraphs, listing all spanning arborescences such that consecutive arborescences differ by reassigning the outgoing edge of exactly one vertex.
Our recursive algorithm generates each spanning arborescence in constant amortized time using $O(n^2)$ space. 
We then extend our algorithm to generate the first known pivot Gray codes for spanning arborescences of split digraphs, achieving constant amortized time per arborescence using $  O(n^2)  $ space. 
Finally, we further generalize our approach to proper interval digraphs. 
When the arborescences are rooted at the leftmost interval, our algorithm generates the pivot Gray code in constant amortized time per arborescence using $  {O}(n^2)  $ space.
\end{abstract}

\section{Introduction}\label{sec:intro}
A \emph{complete digraph} (or \emph{bidirected complete graph}) on $  n  $ vertices, denoted by $  \overrightarrow{K}_n  $, is a directed graph (or \emph{digraph}) in which every pair of distinct vertices is connected by exactly two edges: one in each direction. 
In other words, for every pair of distinct vertices $  u  $ and $  v  $, both edges $u \rightarrow v  $ and $  v\rightarrow u  $ are present.
The \emph{underlying undirected graph} of a directed graph is the undirected graph on the same vertex set that contains an edge between two vertices if and only if there is at least one directed edge between them in either direction.
As such, the underlying undirected graph of $  \overrightarrow{K}_n  $ is the undirected \emph{complete graph} $  K_n  $ (also referred to as a \emph{clique}), in which every pair of distinct vertices is connected by exactly one undirected edge.
For example, the complete digraphs on one to five vertices, that is $  \overrightarrow{K}_1  $ to $  \overrightarrow{K}_5  $ are as follows:
\begin{figure}[H]
    \centering
    \includegraphics[width=0.8\linewidth]{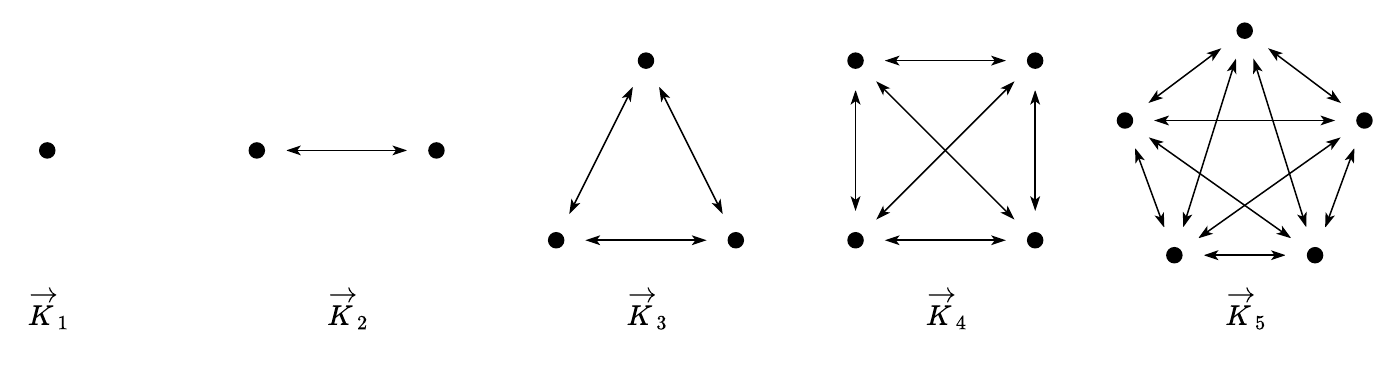}
        \caption{Complete digraphs from $\overrightarrow{K}_1$ to $\overrightarrow{K}_5$.}
    \label{fig:k1k5}
\end{figure}

A \emph{spanning arborescence} of a directed graph $ \overrightarrow{G}  $ is a subgraph that includes all vertices of $  \overrightarrow{G}  $ and forms a directed tree rooted at some vertex $  r  $, with all edges directed towards the root.
In other words, it is a subgraph of $  \overrightarrow{G}  $ in which there is a unique directed path from every vertex to the root, and it contains no directed cycles.
A \emph{spanning tree} of an undirected graph $G$ is a subgraph that includes all vertices of $G$ and forms a tree. 
Throughout this paper, all directed graphs considered in our framework are \emph{symmetric digraphs}. 
A symmetric digraph $  \overrightarrow{G}  $ is obtained from an undirected graph $  G  $ by replacing each undirected edge with a pair of oppositely directed arcs. 
In such symmetric digraphs, every undirected spanning tree admits exactly one valid orientation as an in-arborescence rooted at a fixed vertex $  r  $. 
As such, for the symmetric digraphs studied in this paper, a spanning arborescence of $  \overrightarrow{G}  $ corresponds exactly to a spanning tree of the underlying undirected graph $  G  $.
Figure~\ref{fig:span16} illustrates the 16 spanning arborescences of $\overrightarrow{K}_4$ that rooted at vertex 2. 
\begin{figure}[t]
    \centering
    \includegraphics[width=1\linewidth]{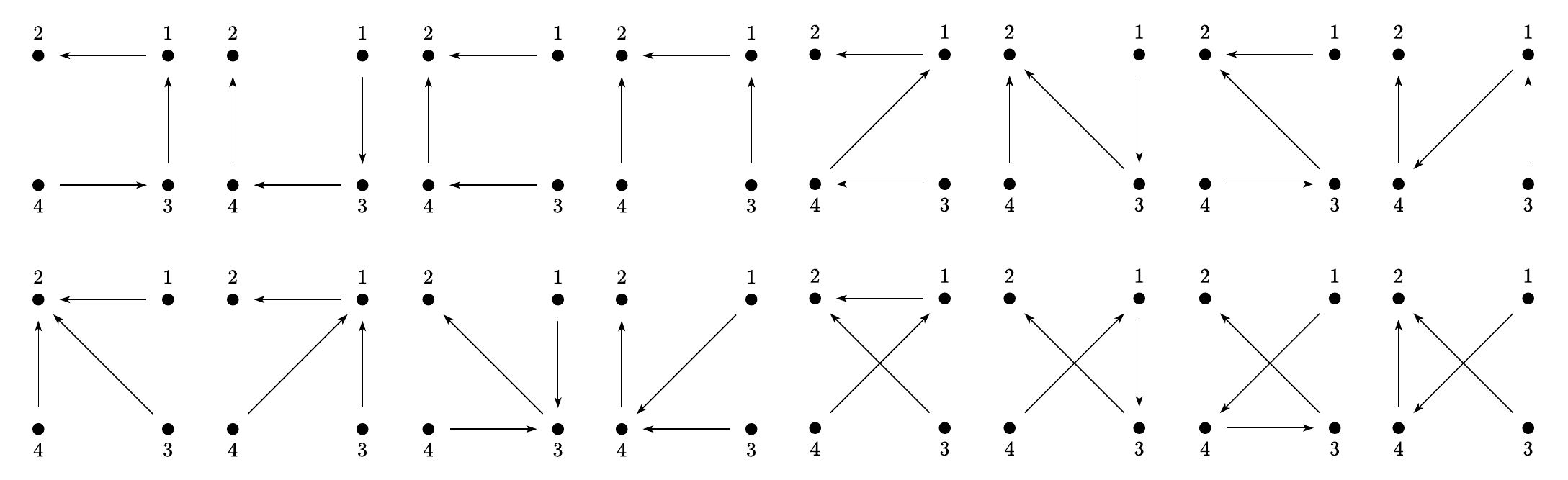}
    \caption{The 16 spanning arborescences of $\overrightarrow{K}_4$ rooted at vertex 2.}
    \label{fig:span16}
\end{figure}

The number of spanning arborescences rooted at a fixed vertex in a complete digraph $\overrightarrow{K}_n$ with $ n $ vertices is given by Cayley’s formula $n^{n-2}$~\cite{Cayley1889}. 
This is the same as the number of spanning trees in the undirected complete graph $  K_n  $, which we can consider the spanning trees of $  K_n  $ as arborescences rooted at a fixed vertex (say, vertex 1).
This enumeration sequence corresponds to sequence A000272 in the Online Encyclopedia of Integer Sequences~\cite{OEIS2025A000272}. 
For example, the number of spanning arborescences for complete digraphs $\overrightarrow{K}_1$ to $\overrightarrow{K}_8$ rooted at a fixed vertex is 1, 1, 3, 16, 125, 1296, 16807, and 262144, respectively.
Cayley provided an informal proof of the formula. 
Later, Pr\"{u}fer~\cite{Prufer1918} and Kirchhoff~\cite{Kirchhoff1847} offered simpler proofs by establishing a one-to-one correspondence between spanning trees (spanning arborescences rooted at a fixed vertex) and Pr\"{u}fer sequences, and by applying the matrix-tree theorem, respectively.
More recently, Liu et al.~\cite{Liu2026} provided a simple one-page proof of Cayley’s formula based on their recursive algorithm for generating spanning trees. 
For more clever proofs of Cayley’s formula, see~\cite{Aigner2004}.

The study of spanning arborescences in directed graphs and their relatives has a long history and has traditionally attracted considerable interest from mathematicians and graph theorists~\cite{Aigner2004, Berger1967, Cayley1889, Chakraborty2019, Gabow1978, Kirchhoff1847,  knuth2011, Mayeda1965, Minty1965, mutze2023combinatorial, Naddef1981, Prufer1918, Rus96, Sav97, Shank1968, Shioura1997, Winter1985}.  
The task of efficiently enumerating all spanning arborescences of digraphs is fundamental in computer science and finds applications across diverse scientific fields including epidemiology, networking, chemistry, biology, astronomy, archaeology, image processing, and social media. 
Notably, in stochastic network analysis, the importance of arborescences is highlighted by the Markov Chain tree theorem originally introduced by Hill in 1966 to model thermodynamic states in biophysics~\cite{hill1966}. 
This theorem expresses the exact stationary distribution of a finite-state directed Markov chain in terms of its spanning arborescences.
For more applications on spanning arborescences of digraphs, see~\cite{bang2008digraphs}.  

One of the most important aspects of combinatorial generation is to list the instances of a combinatorial object so that consecutive instances differ by a specified \emph{closeness condition} involving a constant amount of change. 
Lists of this type are called \emph{Gray codes}. 
This terminology is due to the eponymous \emph{binary reflected Gray code} (BRGC) by Frank Gray~\cite{Gra53}, which orders the $2^n$ binary strings of length $n$ so that consecutive strings differ by one bit. 
For example, when $n = 4$ the order is
$$0000, 1000, 1100, 0100, 0110, 1110, 1010, 0010, 0011, 1011, 1111, 0111, 0101, 1101, 1001, 0001.$$ 
Variations that reverse the entire order or the individual strings are also commonly used
in practice and in the literature. 
For more information about binary reflected Gray code and its variations, see~\cite{10.1007/978-3-030-85088-3_15,flipswap,Vajnovszki2007GrayCO}. 
The BRGC listing is a \emph{1-Gray code} in which consecutive strings differ by one symbol change. 
In this paper, we utilize a generalized version of the binary reflected Gray code to construct Gray codes for spanning arborescences of complete digraphs, split digraphs, and proper interval digraphs.


An interesting problem related to graphs is thus to discover a Gray code for spanning arborescences of digraphs.
To address this, we first recall the corresponding notion in the undirected setting.
A \emph{pivot Gray code} for the spanning trees of an undirected graph is a listing of all spanning trees such that consecutive trees differ by a single pivot operation: the removal of one edge and the addition of another edge that shares a common endpoint.
For complete graphs, the problem of finding a Gray code for spanning trees of complete graphs was posed as an open problem by Knuth in \emph{The Art of Computer Programming, Volume 4} \cite{knuth2011} (Ex. 101 in Sec. 7.2.1.6), with a difficulty rating of 46 out of 50:
\begin{itemize}
\item Is there a simple revolving-door way to list all $n^{n-2}$ spanning trees of the complete graph $K_n$? (The order produced by Algorithm S is quite complicated.) 
\end{itemize}
The \emph{revolving-door} property, as described by Knuth, refers to a Gray code in which consecutive spanning trees differ by the substitution of one edge for another \cite{knuth2011}. 
Smith~\cite{smith1997}, Merino, M\"{u}tze, and Williams~\cite{Merino2022}, and Merino and M\"{u}tze~\cite{doi:10.1137/23M1612019} provided an edge-exchange Gray code for spanning trees of complete graphs. 
Notably, Smith’s algorithm corresponds to the Algorithm S referenced in Knuth’s open problem.
Recently, Liu et al.~\cite{Liu2026} resolved a stronger version of this long-standing open problem of Knuth by constructing the first pivot Gray code for the spanning trees of the complete graph $  K_n  $, which can be generated in constant amortized time using $  O(n^2)  $ space.
For example, Liu et al.'s algorithm produces the pivot Gray code for spanning trees of $K_5$ as shown in Figure~\ref{fig:smith}. 
\begin{figure}[t]
    \centering
    \includegraphics[width=0.9\linewidth, trim={0 0 0 0}, clip]{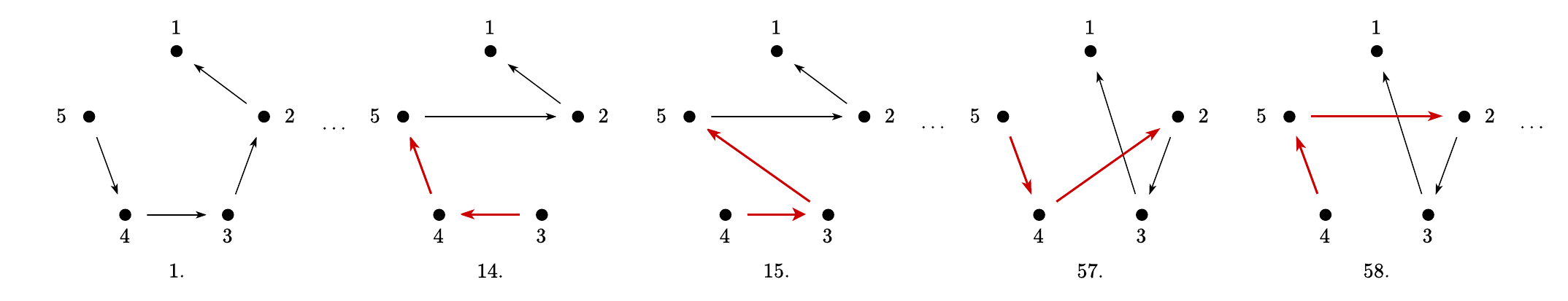}
    \caption{Pivot Gray code for spanning trees of $K_5$ by Liu et al.'s algorithm. Only part of the listing is shown. }
    \label{fig:smith}
\end{figure}

It is important to note, however, that a pivot Gray code in the undirected setting does not directly translate to the directed setting. 
While a single undirected pivot replaces one edge with another incident to the same vertex, the corresponding operation in a digraph may require reversing the orientation of many edges to maintain a valid arborescence rooted at a fixed vertex.
For example, the pivot Gray code generated by Liu et al.’s algorithm requires multiple directed edge changes when applied to digraphs, as the transitions from tree 14 to tree 15, and the transitions from tree 57 to tree 58 in Figure~\ref{fig:smith} for $\overrightarrow{K}_5$ require reversing the orientation of many edges when the arborescences rooted at vertex 1.

We thus require a stricter definition of the pivot operation for spanning arborescences in digraphs. 
Let $  T  $ be a spanning arborescence of a digraph $  \overrightarrow{G}  $ rooted at $  r  $. 
A \emph{directed pivot operation} transforms $  T  $ into a new valid arborescence $  T'  $ by reassigning the unique outgoing edge of exactly one vertex. 
Formally, $  T' = (T \setminus \{u \to p\}) \cup \{u \to q\}  $, where $  u \neq r  $, $  q \neq u  $, $  u \to q \in E(\overrightarrow{G})  $, and $  T'  $ remains acyclic. 
In other words, both the removed and added edges originate from the same vertex $  u  $, which corresponds to updating the parent pointer of a single non-root vertex.
This stricter definition of a pivot change for arborescences in digraphs leads to another open problem regarding spanning arborescence of digraph by Knuth in \emph{The Art of Computer Programming, Volume 4} \cite{knuth2011} (Ex. 102 in Sec. 7.2.1.6), with a difficulty rating of 46 out of 50:
\begin{itemize}
    \item Does every directed graph $  G  $ admit a pivot Gray code on its spanning arborescences rooted at a given vertex?
\end{itemize}
In fact, generating pivot Gray codes for spanning arborescences is a strictly harder problem. 
As established earlier, since we focus on symmetric digraphs, any pivot Gray code for spanning arborescences immediately yields a valid pivot Gray code for the underlying undirected spanning trees (by simply ignoring the edge directions).


While pivot Gray codes for undirected spanning trees are known for several important graph classes (e.g., fan graphs, wheel graphs, outerplanar graphs, and recently complete graphs), the situation for directed arborescences is far less developed.
In a very recent development, Bonamy et al.~\cite{journals/corr/BonamyHLR26} proved the existence of a pivot Gray code for spanning arborescences of tournaments (and hence also of complete digraphs). 
Behroozia and M\"{u}tze~\cite{Behrooznia2025} also proved the existence of a pivot Gray code for spanning arborescences of outerplanar graphs. 
However, their proofs are purely existential and do not yield an efficient generation algorithm.
It is indeed unknown if such a pivot Gray code for spanning arborescences even exists for many other graph classes. 
In this paper, we address this open problem by providing the first known efficient algorithm that generates a pivot Gray code for spanning arborescences of the complete digraph $  \overrightarrow{K}_n  $ in constant amortized time using $  O(n^2)  $ space. 
We then extend our approach to obtain the first known efficient algorithms that generate pivot Gray codes for spanning arborescences of split digraphs and proper interval digraphs rooted at the leftmost interval in constant amortized time per arborescence using $  O(n^2)  $ space. 

Throughout this paper, we are concerned with the generation of spanning arborescences. Since arborescences are defined on directed graphs, all graphs considered in our computational framework are digraphs.
The remainder of the paper is organized as follows. In Section~\ref{sec:pivot}, we present a recursive algorithm for generating a pivot Gray code for the spanning arborescences of complete digraphs. In Section~\ref{sec:split}, we introduce a host–skeleton–routing framework that enables the generation of pivot Gray codes for spanning arborescences of split digraphs. Finally, in Section~\ref{sec:properinterval}, we further extend this framework to generate pivot Gray codes for spanning arborescences of proper interval digraphs rooted at the leftmost interval.

\section{Spanning arborescences of complete digraphs $\overrightarrow{K}_n$}
\label{sec:pivot}

In this section, we present the first known efficient algorithm for generating a pivot Gray code of all spanning arborescences of the complete digraph $\overrightarrow{K}_n$. 
Our algorithm builds upon the level-by-level generation framework originally developed for spanning trees of undirected complete graphs in~\cite{Liu2026}. 
We first review the key ideas of that algorithm. 
We then describe the necessary modifications required to allow us to generate spanning arborescences of $\overrightarrow{K}_n$. 

\subsection{Pivot Gray code construction}\label{sec2:pivotGenS}
\begin{figure}[t]
    \centering
    \includegraphics[width=0.95\linewidth, trim={0 0 0 0}, clip]{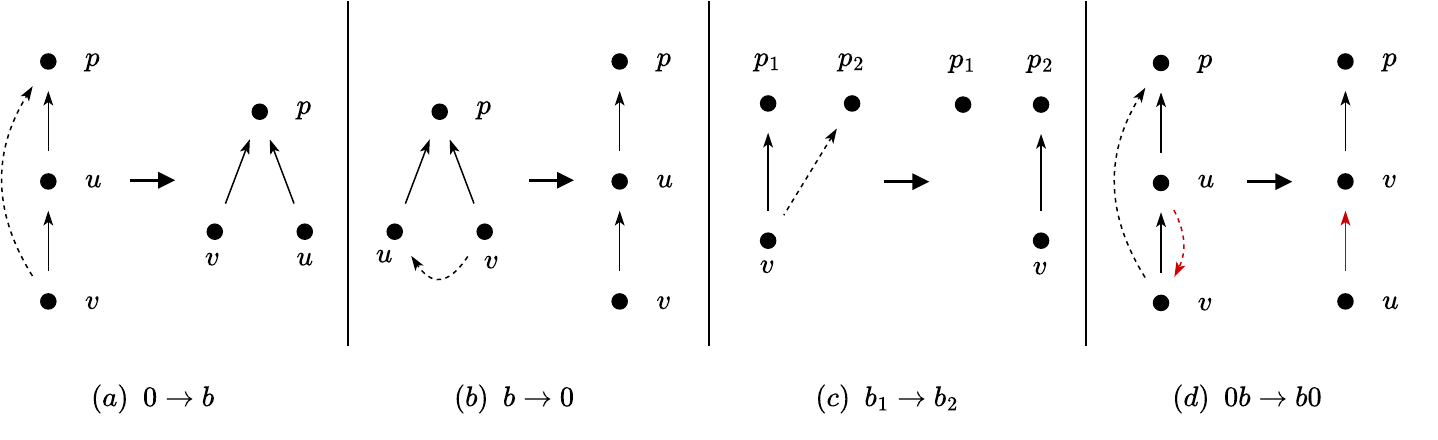}
    \caption{The four cases for a pivot change
on an edge involving a vertex at level $\ell - 1$, while preserving the subtree $S_{\ell-1}$ spanning levels 0 to $\ell - 1$.}
    \label{fig:kn_Gray}
\end{figure}

In the undirected framework in~\cite{Liu2026}, spanning trees are constructed level-by-level and represented by an $h$-tuple $T = (t_1, t_2, \ldots, t_h)$, where $h$ is the height of the tree rooted at vertex 1. 
For each level $  \ell  $, the string $  t_\ell = a_1 a_2 \cdots a_{n-q_{\ell-1}}  $ is a $  k  $-ary string (with $  k = p  $, the number of vertices at level $  \ell-1  $) that encodes the parent connections of the $  n-q_{\ell-1}  $ remaining vertices. 
Here, $  a_i > 0  $ indicates a connection to the $  a_i  $-th vertex at level $  \ell-1  $, while $  a_i = 0  $ indicates that the vertex is not yet connected at this level. 
To ensure connectivity, $  t_\ell  \neq 0^{n-q_{\ell-1}}$.

A pivot Gray code is obtained by recursively traversing all valid tuples $T$ such that consecutive trees differ by a single pivot change. 
This is achieved by iterating the strings $t_\ell$ according to a constrained $k$-ary Gray code sequence for each level $\ell$. 
The algorithm maintains two sets: $  P = \{ v_i \in V \mid a_i > 0 \}  $, consisting of the vertices at level $  \ell-1  $, and $  C = \{ v_i \in V \mid a_i = 0 \}  $, consisting of the vertices not yet placed in levels $0$ through $  \ell-1  $. 
The parent connections between the candidate vertices in $  C  $ and the active parents in $  P  $ are encoded by a $  k  $-ary string $  t_\ell = b_1 b_2 \cdots b_{|C|}  $, where the base $  k = |P| + 1  $ accounts for the $  |P|  $ possible parents in the current level plus the option for a vertex to remain unassigned. 
If $  |C| = 0  $, then all $  n  $ vertices have been placed, so the algorithm outputs the completed spanning tree and returns. 
Otherwise, it initializes $  t_\ell  $ according to the current parent assignments between $  C  $ and $  P  $, and recursively invokes the generation procedure for level $  \ell+1  $. The algorithm then iteratively updates $  t_\ell  $ using a constrained   $k  $-ary strings Gray code generator (excluding the all-zero string to maintain connectivity). 
After each update, the algorithm recurses to level $  \ell+1  $ to exhaustively generate all valid extensions for the deeper levels.
In the undirected setting, the transitions between consecutive $k$-ary strings correspond to four types of pivot changes (cases (a)--(d)) (see Figure~\ref{fig:kn_Gray}). 
The first three cases (cases (a)--(c)) correspond to reassigning a single edge. 
The fourth case (case (d)) occurs when the $k$-ary Gray code generator trying to avoid the all-zero string. 
This case involves a two-symbol swap (a zero and a non-zero symbol), in which a vertex drops a child and connects to one of its descendants.

A natural approach to generalize this framework to complete digraphs is thus to replace the undirected root (vertex 1) with the designated root $r$ of the spanning arborescence, and apply the same level-by-level $  k  $-ary encoding to encode an arborescence. 
For example, Figure~\ref{fig:k5_spanning} shows the spanning arborescence obtained by rooting the spanning tree $  T = (00100, 1100, 22)  $ at vertex $4$, which can be represented by the tuple $T = (11100, 33)$ with the root $  r=4  $. 
While the tuple representation generalizes directly, case (d) transitions might become problematic in the directed setting.
When a vertex $u$ disconnects from $v$ and connects to a descendant $w$ of $v$, the ancestor--descendant relationship between $v$ and $w$ is reversed. 
Maintaining a valid arborescence then requires reversing the directions of all directed edges along the path from $  v  $ to $  w  $, resulting in multiple pivot changes.
For instance, applying case (d) to the arborescences in Figure~\ref{fig:k5_spanning} might attempt to detach vertex 3 from vertex 4 and attach vertex 5 to vertex 4. 
In the directed setting, this would require reversing the edge $5 \to 3$ to preserve the arborescence property, leading to multiple pivot changes, even though the corresponding move is a valid single pivot change in the undirected setting.
Our goal is therefore to eliminate case (d) transitions entirely.

\begin{figure}[t]
\centering
\includegraphics[width=0.9\linewidth]{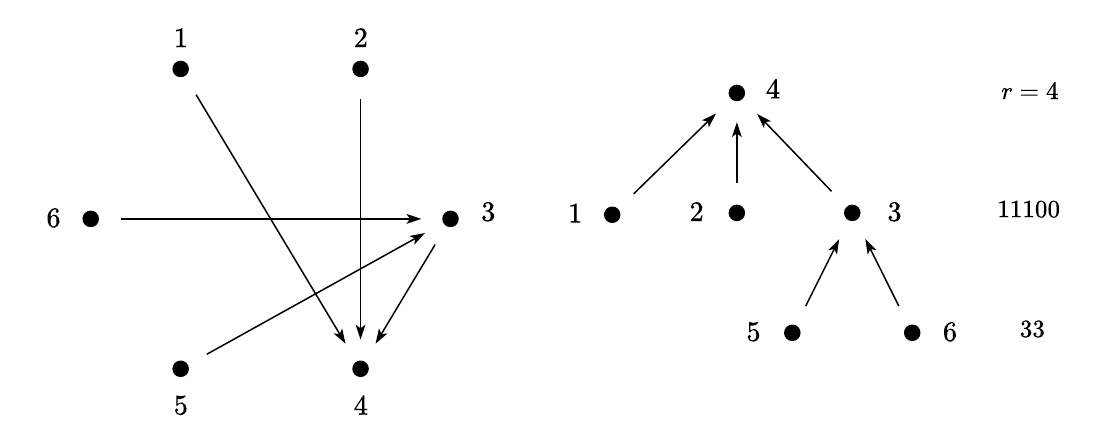}
\caption{The spanning arborescence and its rooted tree of $K_6$ that corresponds to the tuple $T = (11100, 33)$ with the root $r = 4$. }
\label{fig:k5_spanning}
\end{figure}

To avoid case (d), we remove all multi-symbol swaps from the underlying $  k  $-ary string Gray code generator. 
This is equivalent to constructing a strict 1-Gray code over the set of $  k  $-ary strings of length $  m  $ that excludes the all-zero string $0^m$. 
However, such a 1-Gray code does not always exist from every starting string. 
For instance, when $  k=1  $ (binary) and $  m=2  $, the valid strings are $  \{01, 10, 11\}  $. 
Starting at $11$ makes it impossible to visit all strings using only single-symbol changes, since moving between $01$ and $10$ requires a two-bit change. 
In contrast, starting from a string with exactly one non-zero entry (e.g., $01$ or $10$) readily admits a valid sequence such as $01 \to 11 \to 10$.

\begin{algorithm}[t]
    \caption{Generating pivot Gray codes for spanning arborescences of complete digraphs}
    \label{alg:complete_arborescence}
    \begin{algorithmic}[1]
        \Procedure{\sc GenArb}{$V, t_{\ell-1} = a_1 a_2 \cdots a_{|V|}$}
        \State{$P \gets v_i \in V \mathrm{~where~} a_i > 0$}
        \State{$C \gets v_i \in V \mathrm{~where~} a_i = 0$}        \If{$\vert C \vert = 0$}
            \State{Print spanning tree}
            \State{\textbf{return}}
        \EndIf
        
        \State{$t_{\ell} = b_1 b_2 \cdots b_{|C|} \gets 0^{\vert C\vert}$}
        \For{$c_i \in C$}
            \If{$\pi(c_i) \in P$}{$~b_i \gets j $~where~$\pi(c_i) = p_j \in P$}
            \EndIf
        \EndFor
        
        \State{{\sc GenArb}($C, t_{\ell}$)}
        
        \While{{\sc GenLinearK}$(t_{\ell}, \vert P\vert + 1) \ne \emptyset$}
            \State{$t_{\ell}, \delta \gets ${\sc GenLinearK}$(t_{\ell}, \vert P \vert+1)$}
            
            \If{$b_{\delta} > 0$}
                \State{$\pi(c_{\delta}) \gets p_{b_{\delta}} \in P$ {\color{blue}$\quad \quad  \quad  \quad  \quad \quad  \quad  \hspace{0.25em}\triangleright \  $case (a) and (c)}}
            \Else
                \State{$w \gets \text{\sc FindLeaf}(C_{\text{last}})$} {\color{blue}$\quad  \quad \quad  \quad  \hspace{0.5em}\triangleright \  $case (b)}
                \State{$\pi(c_{\delta}) \gets w$}
            \EndIf
            
            \State{{\sc GenArb}($C, t_{\ell}$)}
        \EndWhile
        \EndProcedure        
    \end{algorithmic}
\end{algorithm}



\begin{restatable}{lemma}{constrainedkary}\label{lem:constrained_k_ary}
There exists a 1-Gray code for the $  k  $-ary strings of length $  m  $ excluding the all-zero string $0^m$, starting from any string in which all non-zero digits are distinct, and ending at a string of the form $0^i y 0^{m-i-1}$ where $  i \in \{0, 1, \dots, m-1\}  $ and for some $  y > 0  $.
\end{restatable}
\begin{proof}
When $  k=1  $ (binary case), the statement follows directly from BRGC. 
The BRGC generates a 1-Gray code for all binary strings of length $  m  $. 
It begins with $0^m$ and ends with $0^{m-1}1$. 
Removing the initial all-zero string yields the desired sequence.
For $  k > 1  $, then the proof is based on the reflectable language framework of Li and Sawada~\cite{li2009gray}, with the modification that we dynamically update the last block of the recursion to force it ends with the string $0^i y 0^{m-i-1}$. 
A detailed proof can be found in Appendix A. 
A complete Python implementation of our reflectable Gray code generator is provided in Appendix D.
\end{proof}

Let $  \overrightarrow{K}_n = (V, E)  $ be the complete digraph on vertex set $  V  $. 
For any nonempty subset $R \subseteq V$, an \emph{$R$-forest} on $V$ is a directed acyclic subgraph of $\overrightarrow{K}_n$ with vertex set $V$ such that every vertex $r \in R$ has out-degree zero, and every vertex $v \in V \setminus R$ has out-degree exactly one.
We say that a $  R  $-forest $  F  $ is \emph{linear} if every non-leaf vertex has out-degree exactly one and in-degree at most one. 
Equivalently, $  F  $ consists of a collection of vertex-disjoint directed paths, each rooted at a distinct vertex in $  R  $.
We denote by $  T_{\ge \ell}  $ the forest induced by considering only the vertices and arcs from level $  \ell  $ to the maximum depth $  h  $.
For example, the forest $  T_{\ge 1}  $ corresponding to the spanning arborescence in Figure~\ref{fig:k5_spanning} is given by the tuple $  T_{\ge 1} = (33)  $, which consists of roots $  R = \{1,2,3\}  $ and the subtree rooted at vertex 3 containing vertices 5 and 6.


\begin{lemma}\label{lem:linear_invariant}
If every string $  t_j  $ ($  j \ge \ell  $) in the tuple $  T  $ encoding a spanning arborescence has strictly distinct non-zero digits, then the forest induced by the edges in $  T_{\ge \ell}  $ is a linear $  R  $-forest.
\end{lemma}
\begin{proof}
If all non-zero digits in $  t_j  $ are distinct, then no two vertices at level $  j  $ share the same parent at level $  j-1  $. 
As such, every vertex at level $  j-1  $ has in-degree at most one from vertices at level $  j  $. Since edges only connect adjacent levels, every vertex in the forest induced by the edges in $  T_{\ge \ell}  $ has in-degree at most one. 
Moreover, because $  T_{\ge \ell}  $ is part of a spanning arborescence, every non-root vertex has out-degree exactly one. 
A directed acyclic graph in which every vertex has in-degree at most one and out-degree at most one consists of a collection of vertex-disjoint directed paths. 
Therefore, $  T_{\ge \ell}  $ is a linear $  R  $-forest.
\end{proof}

The idea of linear $R$-forests was first established by M{\"u}tze~\cite{TorstenNotes}, who proved that there exists a pivot Gray code between any two linear $R$-forests on complete graphs, provided that their first-level connection strings are distinct (Theorem 14 in~\cite{TorstenNotes}). 
While M{\"u}tze's result establishes the structural existence of these Gray codes via the Hamilton-connectedness of flip graphs, our primary objective is to design an explicit, constructive algorithmic framework that physically generates each arborescence in strictly constant amortized time.

To this end, our goal is thus to ensure that every recursive call at level $  \ell  $ begins with a linear $  R  $-forest. 
Then by recursively applying Lemma~\ref{lem:constrained_k_ary}, we guarantee that the generation at each level terminates with a string of the form $0^i y 0^{m-i-1}$, which resolves the $  R  $-forest at that level into a single directed path.
To maintain the linear $  R  $-forest invariant across all levels, we refine case (b) from the original undirected framework. 
In the undirected version, case (b) allows a vertex $  v  $ to disconnect from its parent and reconnect to an arbitrary sibling at the same level. 
However, if that sibling already has a child, the reconnection creates an in-degree of two, violating linearity.
In our directed setting, we require that $  v  $ attaches exclusively to a leaf (a vertex with in-degree zero) of a descendant of a sibling. 
Since the recursive generation ensures that each subtree collapses into a directed path by the time the 1-Gray code $k$-ary strings are exhausted, attaching to a leaf simply extends an existing path without introducing branches. 
Applying this leaf-attachment rule recursively guarantees that every intermediate forest remains linear throughout the generation process.
Pseudocode for this spanning arborescences generator is given in Algorithm~\ref{alg:complete_arborescence}.
For the sake of clarity, the algorithm is presented in its conceptual base form. 
It can be optimized to run in strictly constant amortized time per generated arborescence by applying the optimizations described in~\cite{Liu2026}.
To run the algorithm, we make the initial call {\sc GenArb}$(\{1, 2, 3, \ldots, n\}, (10^{n-1}))$ which sets $V = \{1, 2, 3, \ldots, n\}$, and $t_0 = 10^{n-1}$. 
The complete details of the constant amortized time optimization, along with a full Python implementation, are provided in Appendix C.
The spanning arborescences of $\overrightarrow{K}_5$ generated by {\sc GenArb} are listed in Appendix B.

\begin{theorem} \label{thm:genarb_complexity}
The algorithm {\sc GenArb} generates a pivot Gray code for spanning arborescences of complete digraphs $\overrightarrow{K}_n$. 
\end{theorem}

\begin{proof}
The algorithm begins with the linear chain $1 \leftarrow 2 \leftarrow \cdots \leftarrow n$, which is a linear forest. 
At each level $  \ell  $, after updating $  t_\ell  $, the algorithm recursively exhausts all valid extensions for levels $  \ell+1  $ to $  h  $. 
By Lemma~\ref{lem:constrained_k_ary}, these deeper levels always terminate in a linear configuration.
Cases (a) and (c) only reassign a parent pointer at level $  \ell-1  $ and do not affect the structure of the forest $  T_{\ge \ell}  $, so the $R$-forest remains linear. 
For case (b), a vertex $  v  $ at level $  \ell  $ is detached and reattached at a deeper level. 
Since $  T_{\ge \ell}  $ is linear prior to this move, the component containing $  v  $ is a single directed path. 
The algorithm attaches $  v  $ to a leaf of another path, which certainly would not create any vertex of in-degree two. 
Thus, the forest $  T_{\ge \ell}  $ remains linear after the transition. 
By induction, the $R$-forest induced from each level to level $h$ remains linear.

We then prove by contradiction that our algorithm exhaustively generates every spanning arborescence. 
Suppose there exists a valid arborescence encoded by a tuple $  T = (t_1, t_2, \dots, t_h)  $ with a specific root $r$ that is missing from the listing. 
By Lemma~\ref{lem:constrained_k_ary} and {\sc GenLinearK}, the absence of $  T  $ implies that the prefix $  (t_1, t_2, \dots, t_{h-1})  $ never appears in the generation at level $  h-1  $. 
Recursively applying the same argument implies that no spanning arborescence at level zero matching $t_1$. 
However, at level zero, {\sc GenArb} generates a 1-Gray code (via {\sc GenLinearK}) that includes all binary strings except $0^{n-1}$, covering all possible $t_1$ for spanning arborescences of $\overrightarrow{K}_n$, a contradiction.
\end{proof}

\begin{theorem}\label{thm:k_cat}
Spanning arborescences of complete digraphs $\overrightarrow{K}_n$ with n vertices can be generated in pivot Gray code order in constant amortized time per tree, using $O(n^2)$ space.
\end{theorem}

\begin{proof} 
The 1-Gray code for $k$-ary strings can easily be generated in constant amortized time using linear space via the reflectable framework in~\cite{li2009gray}. 
Most of the other operations in {\sc GenArb} can be completed in constant time using linear space, except for case (b), where finding the leaf of a path takes $  O(j)  $ time, where $  j  $ is the length of the path being extended (see Appendix C for implementation details using child pointers). 
However, this cost is amortized: extending the path by $  j  $ levels immediately triggers recursive calls that generate at least one new arborescence at each of the deeper $j$ levels. 
As such, the algorithm runs in constant amortized time per arborescence, using linear space.
\end{proof}

Finally we provide a useful corollary which will become handy when we discuss split digraphs and proper interval digraphs in the subsequent sections. 

\begin{corollary}\label{cor:genarb_linear}
The algorithm {\sc GenArb} begins and ends with a linear spanning arborescence, that is, a single directed path in which every vertex except the root has exactly one parent.
\end{corollary}
\begin{proof}
The algorithm initializes with the linear chain $1 \leftarrow 2 \leftarrow \cdots \leftarrow n$. By Lemma~\ref{lem:constrained_k_ary}, every recursive call terminates with a string of the form $0^i y 0^{m-i-1}$ for some $  y > 0  $, which encodes a single directed path. Therefore, the final arborescence generated by the algorithm is linear.
\end{proof}

\section{Spanning arborescences of split digraphs}\label{sec:split}

This section presents a pivot Gray code for spanning arborescences of split digraphs.
A graph is a \emph{split graph} if its vertex set can be partitioned into a clique $  K  $ and an independent set $  I  $. 
The symmetric \emph{split digraph} is obtained by replacing each undirected edge with a pair of symmetric directed edges. 
For example, Figure~\ref{fig:split_example} shows a split digraph consisting of a clique $  \overrightarrow{K}_5  $ on five vertices and an independent set $I$ with vertices $  \{6,7,8\}  $.
Without loss of generality, we assume that the clique is a maximal clique, that is, no vertex in the independent set is adjacent to every vertex in the clique.
We also assume that the underlying split graph is connected; equivalently, every independent-set vertex has at least one neighbor in the clique.

Split graphs form an important subclass of \emph{chordal graphs}. 
In fact, as the number of vertices $  n\to\infty  $, the proportion of chordal graphs that are split graphs approaches one~\cite{Bender1985}. 
As such, an efficient algorithm for generating spanning arborescences of split digraphs essentially solves the problem for almost all chordal digraphs in the asymptotic sense.
Split graphs have long attracted significant attention due to their rich structural and algorithmic properties~\cite{Golumbic1980, hammer1981splittance, merris2003split, rose1970triangulated, wood2004characterisations}. 
For instance, while fundamental optimization problems like maximum independent set, minimum clique cover, and graph coloring are NP-hard on general graphs, they admit elegant linear-time algorithms on split graphs (and more generally, on chordal graphs) via simplicial elimination orderings~\cite{Gavril1972}. 
Furthermore, split graphs serve as a vital tractability boundary: the cluster deletion problem, which remains NP-complete even on restricted subclasses of chordal graphs, becomes solvable in polynomial time when restricted to split graphs~\cite{Bonomo2015}. 
These tractability results stem from the simple clique--independent set partition, which also yields a clean forbidden--subgraph characterization: split graphs are precisely the graphs containing no induced $  C_4  $, $  C_5  $, or $2K_2$~\cite{Foldes1977}.
Despite these nice properties of split graphs, it was not previously known whether split graphs admit a pivot Gray code for their spanning trees (or spanning arborescences for their digraphs counterpart). 
We resolve this question affirmatively by constructing the first pivot Gray code for spanning arborescences of split digraphs. 
Our construction builds directly upon the framework developed for complete digraphs in the previous section.

\begin{figure}[t]
\centering
\includegraphics[width=0.55\linewidth, trim={0 0 0 0}, clip]{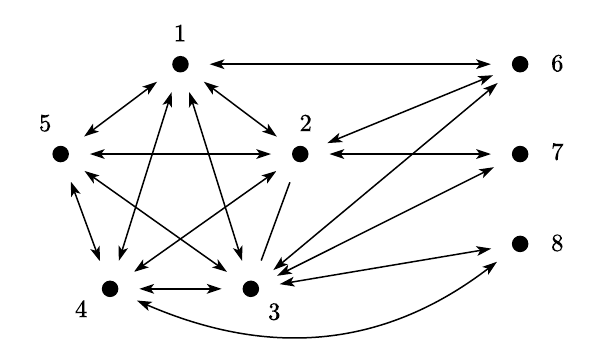}
    \caption{A split digraph with a clique $\overrightarrow{K}_5$ and an independent set containing the vertices $\{6, 7, 8\}$.}
\label{fig:split_example}
\end{figure}

The main idea of our algorithm is to decompose a spanning arborescence of a split digraph $  \overrightarrow{G}=(\overrightarrow{K} \cup I, E)  $ into three combinatorial components: the \emph{host}, the \emph{skeleton}, and the \emph{routing}. 
Let the independent set be ordered as $  I = \{i_1, i_2, \dots, i_m\}  $. 
For each independent vertex $  x \in I  $, we define its set of \emph{out-neighbors} in the clique (i.e., the vertices in $  \overrightarrow{K}  $ to which $  x  $ has a directed edge) as $N(x) = \{ v \in \overrightarrow{K} \mid (x \to v) \in E \}$, where $  E  $ denotes the set of directed edges of the digraph.
We first assume the root $  r  $ lies in the clique $  \overrightarrow{K}  $. 
The case where the root lies in the independent set $  I  $ will be discussed later in this section.
We then demonstrate that every spanning arborescence of a split digraph $  \overrightarrow{G}  $ rooted at a clique vertex $  r\in \overrightarrow{K}  $ can be uniquely represented by a tuple $  S = (r, \rho, T_K, \alpha)  $, where $r$ is the root in the clique $\overrightarrow{K}$, with $\rho$, $T_k$, and $\alpha$ defined as follows.

\paragraph{1. Host ($  \rho  $):}
The {host} mapping determines the outgoing edges from the independent set $  I  $. 
In any spanning arborescence rooted at $  r\in \overrightarrow{K}  $, every vertex $  x\in I  $ has exactly one outgoing edge directed to a vertex in $  \overrightarrow{K}  $. 
The {host} configuration is represented by a parent mapping $  \rho: I \to \overrightarrow{K}  $, which determines the outgoing edges from the independent set $  I  $. 
In any spanning arborescence rooted at $  r \in \overrightarrow{K}  $, every vertex $  x \in I  $ has exactly one outgoing edge directed to a vertex in $  \overrightarrow{K} $, denoted $  \rho_x \in \overrightarrow{K}  $.
We encode the host mapping $  \rho  $ as a mixed-radix string $  h = h_1 h_2 \cdots h_m  $, where the digit $  h_y \in \{0, 1, \dots, |N(i_y)|-1\}  $ is an index indicating that $  \rho(i_y)  $ is the $  (h_y + 1)  $-th smallest vertex in the neighborhood $  N(i_y) \subseteq \overrightarrow{K}  $.

\paragraph{2. Skeleton arborescence ($  T_K  $):}
Since the original structure is a spanning arborescence, $  T_K  $ is itself a spanning arborescence of the complete digraph $  \overrightarrow{K}  $ rooted at $  r  $.
Every vertex $  u\in \overrightarrow{K}\setminus\{r\}  $ has exactly one outgoing edge in the arborescence. 
This edge either points directly to another clique vertex $  v\in \overrightarrow{K}  $ or to an independent vertex $  x\in I  $. In the latter case, the path continues from $  x  $ to $  \rho_{x} = v\in \overrightarrow{K}  $. We define the \emph{skeleton} $  T_K  $ as the projection of these connections onto the clique $  \overrightarrow{K}  $. 
Formally, an edge $  (u \to  v) \in T_K  $ exists if the arborescence contains either the \emph{direct path} $ u \to v$,  or a path $  u\to x\to v  $ for some $  x\in I  $ which we denote as an \emph{alias path}. 

\paragraph{3. Routing ($  \alpha  $):}
While the skeleton $  T_K  $ determines the connectivity within the clique, the \emph{routing} $  \alpha  $ specifies the actual edges used in $T$ to connect the clique vertices. 
For each edge $  u \to v  $ in the skeleton $  T_K  $, vertex $  u  $ may connect to $  v  $ either directly or via an independent vertex acting as an ``alias'' (similar to how a computer network works). For a fixed host $  \rho  $, the set of valid aliases for $  (u \to v) \in T_K  $ is
$$A_\rho(u,v) = \{ x \in I \mid (u \to x)\in E \text{ and } \rho_x = v \}.$$
The routing choices for all $  u\in \overrightarrow{K}\setminus\{r\}  $ are encoded by a mixed-radix string $  \alpha  $, where for each $  u  $ the symbol $  \alpha_u \in \{0,1,\dots,|A_\rho(u,v)|\}  $. If $  \alpha_u=0  $, then $  u  $ connects directly to $  v  $ via the clique edge $  u \to v  $. If $  \alpha_u>0  $, then $  u  $ routes through the $  \alpha_u  $-th smallest vertex in $  A_\rho(u,v)  $.


\begin{lemma}\label{lem:split_bijection}
Let $\overrightarrow{G} = (\overrightarrow{K} \cup I, E)$ be a split digraph with root $r \in \overrightarrow{K}$. There is a bijection between the spanning arborescences of $\overrightarrow{G}$ rooted at $r$ and the valid tuples $S = (r, \rho, T_K, \alpha)$.
\end{lemma}

\begin{proof}
We show that the mapping is both injective and surjective.

\textbf{(Arborescence $\to$ Tuple):} 
Let $T$ be any spanning arborescence of $\overrightarrow{G}$ rooted at $r$. Since $I$ is an independent set, every vertex $x \in I$ has out-degree exactly one, and its unique out-neighbor lies in $\overrightarrow{K}$. This uniquely determines the host string $\rho$. 

For each $u \in \overrightarrow{K} \setminus \{r\}$, the unique outgoing edge from $u$ in $T$ either goes directly to some $v \in \overrightarrow{K}$ or to some $x \in I$ such that $\rho_x = v$. 
In both cases, it induces a unique logical edge $u \to v$ in the skeleton $T_K$. 
Moreover, $  T_K  $ must be acyclic, because any directed cycle in $  T_K  $ would induce a directed cycle in $  T  $, a contradiction.
Thus, $T_K$ is a valid spanning arborescence on $\overrightarrow{K}$.
The specific edge chosen (direct or via alias) uniquely determines the routing symbol $\alpha_u$. 
Thus, each arborescence maps to exactly one valid tuple.

\textbf{(Tuple $\to$ Arborescence):} 
Conversely, consider any valid tuple $(r, \rho, T_K, \alpha)$. The host string $\rho$ assigns exactly one outgoing edge from each vertex in $I$ to a vertex in $\overrightarrow{K}$. The skeleton $T_K$ together with the routing string $\alpha$ assigns exactly one outgoing edge to each vertex in $\overrightarrow{K} \setminus \{r\}$. This construction yields exactly $|\overrightarrow{K}|-1 + |I| = |V|-1$ edges.
Since $T_K$ is a spanning arborescence on $\overrightarrow{K}$ rooted at $r$, every vertex in $\overrightarrow{K}$ has a directed path to $r$. 
Every vertex in $I$ has a path to $r$ because it points to some vertex in $\overrightarrow{K}$ (via $\rho$) and all vertices in $\overrightarrow{K}$ reach $r$. 
Furthermore, every non-root vertex has out-degree exactly one, with all paths directed toward the root. 
The resulting digraph is thus acyclic and is a spanning arborescence of $\overrightarrow{G}$ rooted at $r$. Finally, different tuples produce different sets of edges, so the mapping is injective.
\end{proof}

\begin{figure}[t]
\centering
\includegraphics[width=0.7\linewidth, trim={0 0 0 0}, clip]{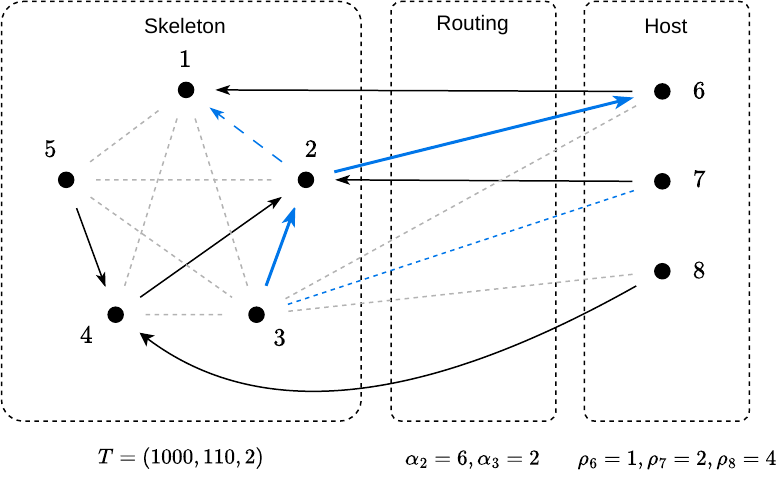}
    \caption{A spanning arborescence of the split digraph in Figure~\ref{fig:split_example} represented by the tuple $  S = (1, 001, (1000,110,2), 1000)  $.}
\label{fig:kn_GfrayR}
\end{figure}

\paragraph{Example of the tuple $S$ encoding (Figure \ref{fig:kn_GfrayR}):}

Consider a split graph with clique $  \overrightarrow{K}=\{1, 2, \dots,5\}  $ and independent set $  I=\{6,7,8\}  $. 
The spanning arborescence shown in Figure~\ref{fig:split_example} can be represented by the tuple $  S = (1, 001, (1000,110,2), 1000)  $ and is interpreted as follows:
\begin{itemize}
\item \textbf{Root:} $  r = 1  $.
\item \textbf{Host ($  \rho = 001  $):} Vertex 6 connects to the first available neighbor in $  \overrightarrow{K}  $ (vertex 1), vertex 7 to its first (vertex 2), and vertex 8 to its second (vertex 4).
\item \textbf{Skeleton ($  T_K = (1000,110,2)  $):} This tuple defines a spanning arborescence on $  \overrightarrow{K}  $ (generated via the complete digraph framework in Section~\ref{sec:pivot}).
\item \textbf{Routing ($  \alpha = 1000  $):} Among the four non-root clique vertices, only vertex 2 uses an alias ($  \alpha_{2}=6  $), routing through its available independent vertex (vertex 6). The remaining vertices use direct clique edges.
\end{itemize}


\subsection{Spanning arborescences rooted at a vertex in the clique}

Since the tuple $ S = (r, \rho, T_K, \alpha) $ uniquely encodes every spanning arborescence rooted at $ r \in \overrightarrow{K} $, generating a pivot Gray code reduces to the problem of traversing the space of valid tuples such that each consecutive tuple differ by a pivot change.
We achieve this with a three-stage nested generator that iterates over the host $  \rho  $, skeleton $  T_k  $, and routing components $  \alpha  $.
\begin{enumerate}
\item \textbf{Host generator ($ \rho $):}
The outermost loop enumerates all valid parent assignments from the independent set $ I $ to the clique $ \overrightarrow{K} $.
Since $ \rho $ is a mixed-radix string, we generate it using a reflectable 1-Gray code described in~\cite{li2009gray}.
Each transition changes the parent of exactly one vertex in $ I $, that is one pivot change.
\item \textbf{Skeleton generator ($ T_K $):}
For a fixed host configuration $ \rho $, the middle layer generates spanning arborescences on the clique $ \overrightarrow{K} $ using the spanning arborescence pivot Gray code generator for complete digraphs (see Section~\ref{sec:pivot} and Algorithm~\ref{alg:complete_arborescence}). 
The fixed alias paths induced by $ \rho $ are treated as valid routing options inside the skeleton generation.
\item \textbf{Routing generator ($ \alpha $):}
For each fixed host $ \rho $ and skeleton $ T_K $, the innermost layer generates the actual routing choices using another reflectable mixed-radix 1-Gray code.
Each step changes the routing decision (direct edge vs. alias) for exactly one clique vertex, that is one pivot change.
\end{enumerate}
Pseudocode of this three-stage nested generator is given in Algorithm~\ref{alg:split_arb}.  
The algorithm {\sc GenSplitArb} begins by initializing the baseline routing options for each clique vertex based on the edges between the clique $  \overrightarrow{K}  $ and the independent set $  I  $. 
The outermost generator, {\sc GenHost}, iterates over all host configurations $  \rho  $, dynamically updating the available routing options after each host change. 
Then for each fixed host configuration, it invokes the middle-layer skeleton generator {\sc GenArb} to enumerate spanning arborescences on the clique $  \overrightarrow{K}  $. 
For each generated skeleton $  T_K  $, the function {\sc GetRoutingMixedRadix} computes the valid routing aliases. 
Finally, the innermost generator {\sc GenRouting} cycles through these routing options. 
At the deepest level of recursion, the tuple $  (r, \rho, T_K, \alpha)  $ is decoded into a spanning arborescence of the split digraph.


Note that, however, na\"{i}ve nesting of this algorithm might create cycles when we finish generating the routing choices $  \alpha  $.
When the host generator is about to change the parent of an independent vertex $ x $, some clique vertices (for example, $u$) may currently use $ x $ as an alias.
If the new parent of $  x $ is one of those clique vertices, the update would introduce a directed cycle (e.g., $  u \to x \to u  $), thereby violating the arborescence property.
To prevent this, we must ensure that immediately before updating $  x  $, no clique vertex $  u \in \overrightarrow{K}  $ is using $  x  $ as a routing alias. 
We achieve this through a \emph{look ahead} mechanism that identifies the upcoming host independent set vertex that is going to be updated, and applies \emph{endpoint control} in the inner generators. 
This is achieved by modifying the {\sc GenHost} function to explicitly identify the independent vertex $  x  $ scheduled for the next pivot and passes it as a parameter to the skeleton generator {\sc GenArb}, which enables the routing generator to safely orient its terminal states, as formalized in the following lemma.


\begin{algorithm}[t]
    \caption{Generating pivot Gray codes for spanning arborescences of split digraphs}
        \label{alg:split_arb}
    \begin{algorithmic}[1]
        \Procedure{GenSplitArb}{$\overrightarrow{K}, I$}
            \State{$R \gets \text{Initialization}(\overrightarrow{K}, I)$}
            \For{$\rho, x \in \text{\sc GenHost}()$} 
                \State{Update $R$}
                \For{$T_k \in \text{\sc GenArb}(\overrightarrow{K}, t_0, x, \text{True})$}
                    \State{$A \gets \text{GetRoutingMixedRadix}(\rho, T_k)$}
                    \For{$\alpha \in \text{\sc GenRouting}(A)$}
                        \State \textbf{yield} decode arborescence $S = (r, \rho, T_K, \alpha)$
                    \EndFor
                \EndFor
            \EndFor
        \EndProcedure
    \end{algorithmic}
\end{algorithm}

\begin{lemma}\label{thm:split_kary}
Let $  x \in I  $ be the independent vertex scheduled to pivot next, and let $  T_K^{\mathrm{end}}  $ be the final skeleton arborescence generated in the current host block. 
If $  T_K^{\mathrm{end}}  $ is a single directed path rooted at $  r  $, then 
regardless of the current routing configuration, the nested generation of the skeleton $  T_K  $ and routing aliases $  \alpha  $ can be dynamically controlled to terminate in a configuration where no vertex in $  \overrightarrow{K}  $ utilizes $  x  $ as an alias (i.e., $  \alpha_u \neq x  $ for all $  u \in \overrightarrow{K}  $).
\end{lemma} 

\begin{proof}
Let $  p \in \overrightarrow{K}  $ be the current host of $  x  $. 
Since $  T_K^{\mathrm{end}}  $ is a single path, every vertex has at most one child. 
In particular, $  p  $ has at most one child in $  T_K^{\mathrm{end}}  $, which we denote as $  u  $. Since a clique vertex can only route through $  x  $ if its skeleton parent is $  p  $, $  u  $ is the only vertex in the entire clique $  \overrightarrow{K}  $ capable of using $  x  $ as an alias at the end of the generation of the skeleton arborescence.
If no such $  u  $ exists (either $  p  $ is a leaf or the edge $  u \to x  $ does not exist), then $  x  $ is unavailable as an alias for all $  u \in \overrightarrow{K}  $, and the claim holds immediately.

Now assume $  u  $ exists. 
Its available routing options contain both the direct choice $  p  $ and the alias $  x  $. Our goal is to force the terminal routing value of $  u  $ to be the direct choice $  p  $, ensuring that the subsequent update of the host vertex $  x  $ does not disrupt the skeleton arborescence.
We prove that the nested generation of $  T_K  $ always offers an opportunity to reset the parity of the reflectable routing generator, safely setting $  \alpha_u  $ to the direct path $  p  $ at the end of the generation of $  T_K  $.

If $\overrightarrow{K}$ has only two vertices $(u, p)$ and an $x$ exists in the independent set allowing $u \to x \to p$, $\{u, x, p\}$ can form a larger clique, violating the assumption that $\overrightarrow{K}$ is maximal. 
Now consider $\overrightarrow{K}$ has at least three vertices. Such a reset opportunity must exist as the last arborescence includes $u \to p$ and there exists an arborescence without $u \to p$.
One of the reset opportunities arises when the skeleton edge $  u \to p  $ is going to be established for the final time prior to $  T_K^{\mathrm{end}}  $ (meaning the path $  u \to p  $ will remain connected until the skeleton block concludes).
Such an opportunity always exists since the edge $u \to p$ is present in the terminal arborescence $T_K^{\mathrm{end}}$. 
At this point, $  u  $ has at least two routing choices: the direct path $  u \to p  $ or the alias path $  u \to x \to p  $.
The routing generator $  \alpha  $ operates as a reflectable mixed-radix Gray code. 
Thus, for every subsequent skeleton arborescence generated, the routing cycle for $  u  $ will be reflected: $  \alpha_u = p  $ becomes $  \alpha_u = x  $, or vice versa.
The final state of $  \alpha_u  $ is therefore completely determined by its initial state and the parity of the number of remaining skeleton configurations, once we place $  \alpha_u  $ at the first level of the reflectable recursion tree of the routing generator.
Therefore, if we know how many arborescences we still need to generate (this value can be computed by exhaustive search, look ahead, or precomputation), then we can set $  \alpha_u  $ to the direct path $  p  $ if the parity of the remaining sequence is even, and to the alias $  x  $ (or another extreme) if the parity is odd. 
As such, as the algorithm exhausts the remaining arborescences in the skeleton block, the reflectable routing generator completes its sequence and $\alpha_u$ comes to rest on the direct path $  p  $.
This way, when we reach $  T_K^{\mathrm{end}}  $, $  \alpha_u  $ will always use the direct path $  p  $, and we can safely pivot the host vertex $  x  $.
\end{proof}

The next natural question is how to determine when the skeleton edge $  u \to p  $ is established for the final time prior to reaching $  T_K^{\mathrm{end}}  $, and how to efficiently compute the parity of the remaining arborescences. 
It turns out that, by exploiting the linear property of $  T_K^{\mathrm{end}}  $, both values can be computed efficiently.

\begin{lemma}\label{lem:endpoint_cat}
If $  T_K^{\mathrm{end}}  $ is a single directed path rooted at $r$, then the exact step when the skeleton edge $u \to p$ is established for the final time prior to reaching $  T_K^{\mathrm{end}}  $ and the parity of the remaining skeleton arborescences can be identified and computed in constant time.
\end{lemma}

\begin{proof}
Since $  T_K^{\mathrm{end}}  $ is a single directed path, the string $  t_\ell  $ at each level $  \ell  $ of the skeleton ends with a single-child configuration of the form $0^a10^b$. 
Without loss of generality, assume the final string at level $  \ell  $ is $0^{c-1}1$, where $  c = |C|  $ is the number of unassigned vertices at that level.
We now trace how the final skeleton edge $  u \to p  $ is established. 

The vertex $  p  $ is a parent passed from the previous level (or $p$ is the root). 
We first analyze the general case where $c \ge 2$.
Since $  T_K^{\mathrm{end}}  $ is a single directed path, level $  \ell-1  $ contains only the vertex $  p  $. 
Thus, the 1-Gray code generator uses BRGC to exhaustively generate the valid constrained binary strings of length $  c  $, excluding $0^c$.
The digit $1$ in the terminal string $  t_\ell = 0^{c-1}1  $ corresponds to the edge $  u \to p  $. 
The final three strings of the BRGC sequence are $0^{c-2}10$, $0^{c-2}11$, and $0^{c-1}1$. 
Therefore, the edge $  u \to p  $ is established for the final time at the transition to $  t_\ell = 0^{c-2}11  $, which can be identified by observing that exactly two strings remain in the BRGC sequence.

We then compute the total number of remaining skeleton arborescences. 
The number of remaining skeleton arborescences correspond to the spanning forests on $  c  $ vertices extending the states $0^{c-2}11$ ($  k=2  $ roots) and $0^{c-1}1$ ($  k=1  $ root). 
By the generalized Cayley's formula for forests, the number of spanning forests on $  c  $ vertices with $  k  $ specified roots is $  k \cdot c^{c-k-1}  $. 
Thus, the total number of remaining skeleton arborescences is
$$2c^{c-3} + c^{c-2} = c^{c-3}(2 + c).$$
Taking this modulo 2, the parity of the remaining arborescences simplifies to $  c \pmod{2}  $. 
Finally, when $  c=1  $, the local sequence reduces to the single-state string $  t_\ell = 1  $, indicating that we have reached $  T_K^{\mathrm{end}}  $. 
Thus, the parity of the remaining global sequence is exactly the parity of $  c  $ in both cases, which can be computed in constant time.
\end{proof}

\begin{algorithm}[!t]
    \caption{Generating pivot Gray codes with routing alias control}
    \label{alg:genarb_endpoint}
    \begin{algorithmic}[1]
        \Procedure{\sc GenArb}{$V, t_{\ell-1}, x, \text{last}$}
            \State{$P \gets v_i \in V \mathrm{~where~} a_i > 0$}
            \State{$C \gets v_i \in V \mathrm{~where~} a_i = 0$}
            \If{$\vert C \vert = 0$} \textbf{yield} \EndIf
            
            \State{$t_{\ell} = b_1 b_2 \cdots b_{|C|} \gets 0^{\vert C\vert}$}
            \State \text{Initialize parents mapping to $P$}

            \State{$n_{\text{str}} \gets \vert P \vert ^{\vert C \vert} - 1$}
            \State{{\sc GenArb}($C, t_{\ell},  x, \text{last} \cap n_{\text{str}} = 1$)}
            \State{$n_{\text{str}} \gets n_{\text{str}} - 1$}
            
            \While{{\sc GenLinearK}($t_{\ell}, \vert P\vert + 1) \ne \emptyset$}
                \State{$t_{\ell}, \delta \gets ${\sc GenLinearK}$(t_{\ell}, \vert P \vert+1)$}
                \State{$c \gets C[\delta]$} 
                
                \If{$b_{\delta} > 0$}
                    \State{$\pi(c) \gets p_{b_{\delta}} \in P$} 
                \Else
                    \State{$w \gets \text{\sc FindLeaf}(C_{\text{last}})$}
                    \State{$\pi(c) \gets w$}
                \EndIf

                \If{$\text{last}$ \textbf{and} $n_{\text{str}}= 2$ \textbf{and} $\pi(x) = p_0$} \hfill\tikzmark{brace-start}
                    \If{$\vert C \vert$ is even}
                        \State{$\alpha_{c} \gets p = R[p][c]_{\text{first}}$}
                    \Else
                        \State{$\alpha_{c} \gets R[p][c]_{\text{last}}$}
                    \EndIf
                \ElsIf{$\text{last}$ \textbf{and} $\vert C \vert = 2$ \textbf{and} $n_{\text{str}}= 1$ \textbf{and} $\pi(x) = w$}
                    \State{$\alpha_{c} \gets R[w][c]_{\text{last}}$} \hfill\tikzmark{brace-end}

                \Else
                    \State{$\alpha_{c} \gets \pi(c) = R[\pi(c)][c]_{\text{first}}$}
                
                \EndIf
                
                \State{{\sc GenArb}($C, t_{\ell}, x, \text{last} \cap n_{\text{str}} = 1$)}
                \State{$n_{\text{str}} \gets n_{\text{str}} - 1$}
            \EndWhile
        \EndProcedure
    \end{algorithmic}
\noindent\makebox[0pt][l]{%
\tikz[remember picture,overlay]{
\draw[
    decorate,
    decoration={brace, amplitude=5pt},
    thick,
    blue
]
([xshift=-11.0em,yshift=2.0ex]pic cs:brace-start)
--
node[midway,right=8pt,blue]{\small endpoint control}
([xshift=-11.0em,yshift=-2.0ex]pic cs:brace-end);
}%
}

\end{algorithm}

Pseudocode for the updated spanning arborescence generator is provided in Algorithm~\ref{alg:genarb_endpoint}. 
The procedure extends Algorithm~\ref{alg:complete_arborescence} (for complete digraphs) by incorporating dynamic routing assignments and endpoint control required for split digraphs.
To track the precise timing of block boundaries, the algorithm uses two state variables: a boolean \emph{last} flag and a countdown variable $  n_{\text{str}}  $. 
The \emph{last} flag is propagated down the recursion tree and becomes true when all previous recursive levels have reached their final configurations. 
The variable $  n_{\text{str}}  $ tracks the number of remaining skeleton configurations at the current level.
By monitoring these variables, the algorithm identifies the exact transition at which the terminal skeleton edge $  u \to p  $ is established. 
The routing alias initialization is then overridden when \emph{last} is true and $  n_{\text{str}}  $ equals 2 or 1. 
At these points, the algorithm detects the formation of $  u \to p  $. 
We then assign $  \alpha_u  $ to the first level of the routing generator’s reflectable recursion tree and initialize the coordinate $  \alpha_u  $ according to the parity of $  c  $.
To generate the complete sequence of spanning arborescences for a split digraph, we make the initial call {\sc GenSplitArb}$(\overrightarrow{K}, I)$. 
This procedure then coordinates the nested traversal by invoking {\sc GenHost}, the updated {\sc GenArb}, and {\sc GenRouting} to resolve the alias assignments. 
A complete Python implementation of the algorithm is provided in Appendix~E. 

As an example, the 21 spanning arborescences of the split digraph from Figure~\ref{fig:split_example}, rooted at vertex 1, are listed in Figure~\ref{fig:split_gray_code_listing}.
Endpoint control is applied at the underlined arborescences 10 and 16 to force the active clique vertices to use their direct skeleton edges immediately before the host transitions of the independent vertices.
In the case of Tree 10, the \emph{last} flag is set to {true} and $  n_{\text{str}} = 2  $. 
Here, the independent vertex scheduled to pivot its host is vertex 4, and the skeleton generator pivots the edge $2 \to 5$ to $2 \to 1$ (i.e., $  u=2  $ and $  p=1  $). 
At this point, as the skeleton edge $2\to 1$ of $  T_K^{\mathrm{end}}  $ is being established, vertex 2 can either connect directly to its skeleton parent 1 or route through the alias path via the independent vertex 4.
Since $  n_{\text{str}} = 2  $ indicates that exactly two skeleton configurations remain in the current block, the algorithm initializes the routing state of vertex $2$ to the direct path $2 \to 1$. 
The reflectable routing generator then traverses the available routing choices, pivoting from the direct path $2 \to 1$ to the alias path $2 \to 4$ (Tree 11). 
In the next step, the skeleton generator transitions to the final arborescence of the block by pivoting vertex 3 to vertex 2 (Tree 12). 
At this stage, the routing generator is reflected, causing the active path to pivot from the alias $2 \to 4$ back to the direct path $2 \to 1$ (Tree 13).
As such, no clique vertex utilizes vertex 4 as an alias, allowing the independent vertex 4 to safely change its host to vertex 2 in the subsequent step (Tree 14).
Then in the case of Tree 16, the \emph{last} flag is set to {true} and $  n_{\text{str}} = 1  $. 
Here, the independent vertex scheduled to pivot its host is vertex 5, and the skeleton generator pivots the edge $2 \to 1$ to $2 \to 3$ (i.e., $  u=2  $ and $  p=3  $). 
At this point, the skeleton edge $2\to 3$ of $  T_K^{\mathrm{end}}  $ is being established, vertex 2 can either connect directly to its skeleton parent 3 or route through the alias path via the independent vertex 5. 
Since $  n_{\text{str}} = 1  $ indicates that only one skeleton configuration remains in the current block, the algorithm initializes the routing state of vertex $2$ to the alias path $2 \to 5$. 
The reflectable routing generator then traverses the available routing choices, pivoting from the alias path $2 \to 5$ to the direct path $2 \to 3$ (Tree 17). 
As such, no clique vertex utilizes vertex 5 as an alias, allowing the independent vertex 5 to safely change its host to vertex 2 in the subsequent step (Tree 18).

\begin{theorem}\label{thm:gensplit_gray}
The algorithm {\sc GenSplitArb} generates a pivot Gray code for the spanning arborescences of a split graph $  \overrightarrow{G} = (\overrightarrow{K} \cup I, E)  $ rooted at a vertex $  r \in \overrightarrow{K}  $.
\end{theorem}

\begin{proof}
By Lemma~\ref{lem:split_bijection}, there is a bijection between the spanning arborescences of $  \overrightarrow{G}  $ rooted at $  r \in \overrightarrow{K}  $ and the valid configuration tuples $  S = (r, \rho, T_K, \alpha)  $. 
The algorithm decomposes the generation into three nested layers corresponding exactly to these components: the host assignments $  \rho  $ for vertices in $  I  $, the skeleton arborescence $  T_K  $ on the clique $  \overrightarrow{K}  $, and the routing aliases $  \alpha  $.
The host assignments $  \rho  $ and the skeleton $  T_K  $ are strictly independent; the structure of $  T_K  $ imposes no restrictions on $  \rho  $. 
Furthermore, while the available option for the routing aliases $  \alpha  $ depends on the current host configuration $  \rho  $, the specific routing options made by the innermost loop do not affect the validity of $  \rho  $ or $  T_K  $. 
Since these routing options are independent and each nested generator exhaustively enumerates every valid string in its domain exactly once, the overall procedure generates every valid tuple $  S  $ exactly once. 
As such, it produces every spanning arborescence of $  \overrightarrow{G}  $ exactly once.

\begin{figure}[t]
\centering
\includegraphics[width=1\linewidth, trim={0 0 0 0}, clip]{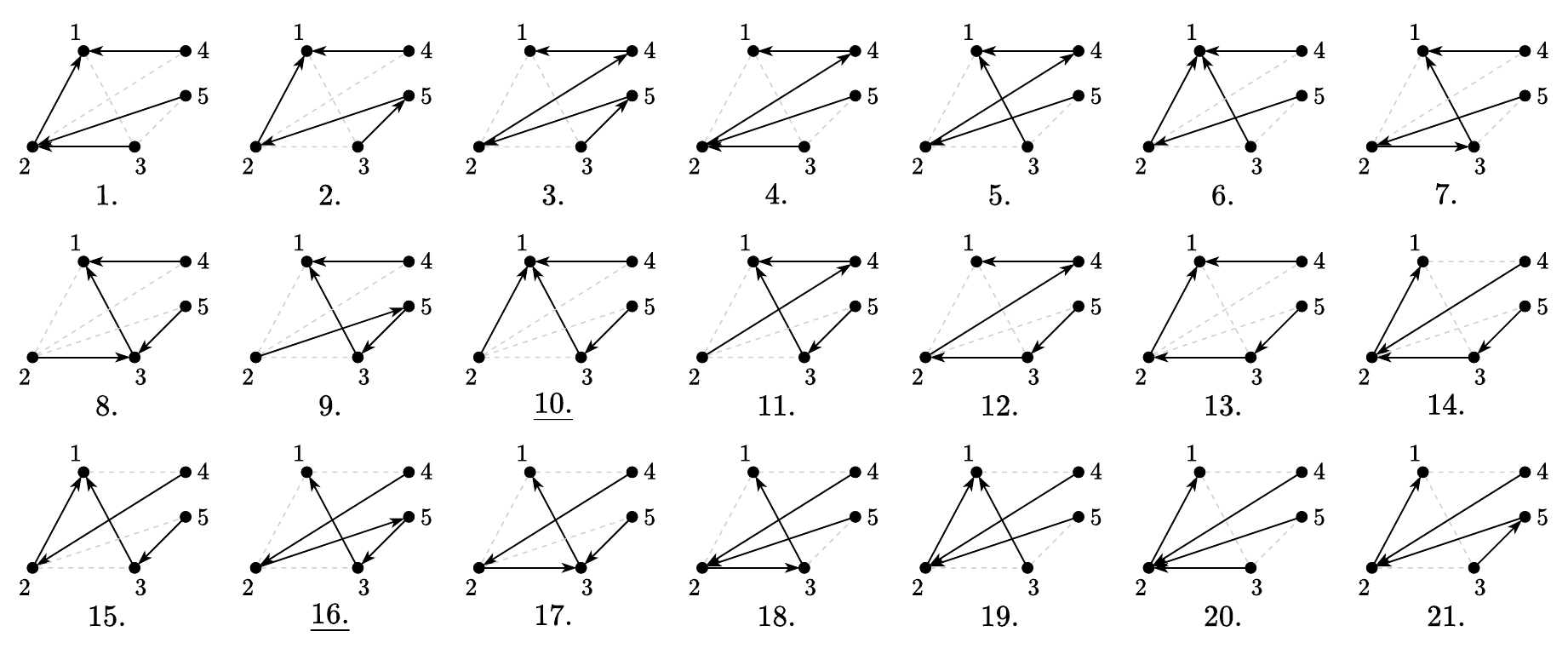}
    \caption{The pivot Gray code sequence of spanning arborescences for the split digraph in Figure~\ref{fig:split_example}, rooted at vertex 1. Endpoint control is applied at the underlined arborescences 10 and 16 to force the clique vertices to use their direct skeleton edges immediately before the host transitions of the independent vertices. }
\label{fig:split_gray_code_listing}
\end{figure}

Within each individual generator (host, skeleton, and routing), the traversal is a 1-Gray code, which translates to a pivot change. 
For both the routing generator and the skeleton generator, each step translates directly to a valid single-edge pivot change and introduces no cycles, thereby preserving the arborescence property.
The only potential threat to the arborescence property occurs at host block boundaries, when the outer host generator changes the parent of an independent vertex $  x \in I  $. 
If some clique vertex were using $  x  $ as a routing alias at that moment, multiple parent pointers would change simultaneously.
However, by Lemma~\ref{thm:split_kary}, the nested generators employ endpoint control to ensure that the routing configuration terminates in a state where no vertex $  u \in \overrightarrow{K}  $ uses $  x  $ as an alias ($  \alpha_u \neq x  $ for all $  u \in \overrightarrow{K}  $). 
Since the routing is safely using the direct skeleton edges, the host transition modifies only the single outgoing edge from $  x $. 
Thus, both internal transitions and boundary transitions correspond to a valid single pivot change.
Since the algorithm visits every spanning arborescence exactly once and every consecutive pair differs by a single pivot change, the resulting listing is a pivot Gray code.
\end{proof}

\begin{theorem} \label{thm:split_complexity}
Spanning arborescences of split graphs can be generated in a pivot Gray code order in constant amortized time per arborescence, utilizing $O(n^2)$ space when rooted at a vertex in the clique. 
\end{theorem}

\begin{proof}
By Theorem~\ref{thm:k_cat}, spanning arborescences of $  \overrightarrow{K}  $ can be generated in constant amortized time per arborescence using $  {O}(n^2)  $ space. 
The reflectable mixed-radix Gray codes used for the host assignments ($  \rho  $) and routing aliases ($  \alpha  $) can be generated in constant amortized time per string.
The only operation that may exceed constant time is thus the transition at the host level. 
When an independent vertex changes its host in $  \overrightarrow{K}  $, the algorithm updates the routing mapping tables (i.e., the available aliases for clique vertices). 
This update takes $  {O}(|\overrightarrow{K}|)  $ time by maintaining a global routing table $  R  $ throughout the execution, and updates only the aliases associated with the neighborhood of the moving vertex. 
However, once the update is complete, the inner skeleton generator produces at least $  |\overrightarrow{K}|^{|\overrightarrow{K}|-2}  $ spanning arborescences, which amortizes the hosting update time to constant time per generated arborescence.
The endpoint control mechanisms that include checking alias availability and computing the parity of the remaining skeleton arborescences can be performed in constant time by Lemma~\ref{lem:endpoint_cat}. 
Finally, for space complexity, the algorithm maintains adjacency mappings, parent pointers, and alias routing tables, which can all be maintained using at most $  {O}(n^2)  $ space. 
Therefore, the algorithm achieves constant amortized time per arborescence, using $  {O}(n^2)  $ space.
\end{proof}

\subsection{Spanning arborescences rooted in the independent set}

We now consider the case where the designated root $  r  $ belongs to the independent set $  I  $.
The key difference is that $  r  $ is adjacent only to a subset $  S = \{ s \in \overrightarrow{K} \mid (s \to r) \in E \} \subseteq \overrightarrow{K}  $ of the clique. 
In any valid spanning arborescence rooted at $  r  $, all paths to the root must pass through $  S  $. 
As such, the set of vertices in $  \overrightarrow{K}  $ that choose $  r  $ as their parent must form a non-empty subset $  R \subseteq S  $.
Once $  R  $ is fixed, the remaining vertices $  \overrightarrow{K} \setminus R  $ induce a spanning forest rooted at $  R  $. 
Since $  \overrightarrow{K}  $ is a clique, this subproblem reduces to the complete-digraph spanning arborescence generation framework.
We handle this structure by treating $  r  $ as the root of the skeleton arborescence $  T_K  $, while restricting its possible incoming edges. 
In the standard complete digraph setting, the first level string $  t_1  $ has length $  |\overrightarrow{K}|-1  $. 
When $  r \in I  $, we instead use a binary string $  t_1  $ of length $  |S|  $, where each bit indicates whether the corresponding vertex in $  S  $ points directly to $  r  $. 
To ensure connectivity, we similarly require that $  t_1 \neq 0^{|S|}  $.
Since $  r  $ belongs to the independent set, it requires no routing decision. 
After fixing $  t_1  $, the generation of the remaining skeleton edges within $  \overrightarrow{K}  $ and the parent assignments for the remaining independent vertices $  I \setminus \{r\}  $ proceeds identically to the clique-rooted split graph case.
Since the restricted $  t_1  $ sequence is generated using BRGC, all transitions correspond to single pivot changes.
As an example, the split graph in Figure~\ref{fig:split_example} can be decomposed into the three components (host, skeleton, and routing) as shown in Figure~\ref{fig:iroot} when rooted at the independent vertex 6.
A complete Python implementation of the updated algorithm {\sc GenSplitArb} for the case where the root lies in the independent set is provided in Appendix~E.

\begin{figure}[t]
\centering
\includegraphics[width=0.65\linewidth, trim={0 0 0 0}, clip]{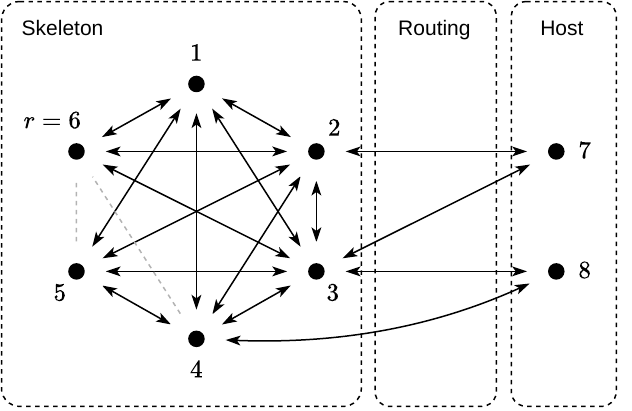}
\caption{The host, skeleton, and routing components of the split digraph from Figure~\ref{fig:split_example} when rooted at the independent vertex 6. Note that the independent vertex 6 is now treated as the root of the skeleton block. Since the root $6 \in I$ is non-adjacent to the clique vertices $4$ and $5$, the first-level connection string $  t_1  $ of the skeleton arborescence $  T_K  $ is restricted to the neighborhood $  S = \overrightarrow{K} \setminus \{4,5\}  $.}
\label{fig:iroot}
\end{figure}


\begin{theorem}\label{thm:gensplitI_gray}
The modified algorithm {\sc GenSplitArb}  generates a pivot Gray code for the spanning arborescences of a split graph $  \overrightarrow{G} = (\overrightarrow{K} \cup I, E)  $ rooted at a vertex $  r \in I  $.
\end{theorem}

\begin{proof}
Any spanning arborescence rooted at $  r \in I  $ must have a non-empty set of incoming edges from a subset $  R \subseteq S \subseteq \overrightarrow{K}  $, where $  S  $ is the set of neighbors of $  r  $ in the clique. 
The remaining vertices in $  \overrightarrow{K} \setminus R  $ induce a spanning forest rooted at $  R  $.
The modified algorithm {\sc GenSplitArb} thus captures this structure by replacing the standard first level string $  t_1 \in \{0,1\}^{|S|} \setminus \{0^{|S|}\}  $, where each bit indicates whether the corresponding vertex in $  S  $ points directly to $  r  $. 
Since the algorithm exhaustively enumerates all valid choices of $  t_1  $ and pairs them with the complete generation of the remaining skeleton forest on $  \overrightarrow{K} \setminus R  $ together with the host and routing assignments for $  I \setminus \{r\}  $, it produces every spanning arborescence exactly once.

Within the modified base level, $  t_1  $ is generated using BRGC excluding the all-zero string. 
Thus, consecutive configurations of $  t_1  $ differ by exactly one bit, corresponding to a single pivot change. 
This modification at the root level does not interfere with the endpoint control mechanisms used for the remaining independent vertices, since $  r  $ cannot serve as a host or routing alias.
The analysis of the rest of the host and routing transitions follows identically to that of Theorem~\ref{thm:gensplit_gray}. 
Therefore, consecutive arborescences in the generated sequence differ by exactly one pivot change, yielding a pivot Gray code.
\end{proof}

\begin{theorem} \label{thm:splitI_complexity}
Spanning arborescences of split graphs can be generated in a pivot Gray code order in constant amortized time per arborescence, utilizing $O(n^2)$ space when rooted at a vertex in the independent set. 
\end{theorem}

\begin{proof}
The modified base level $  t_1  $, which encodes the subset of directed edges from the neighborhood $  S \subseteq \overrightarrow{K}  $ to the root $  r \in I  $, is generated using BRGC which can be done in constant amortized time per string.
The remainder of the algorithm operates identically to the clique-rooted case and thus the algorithm generates each spanning arborescence in constant amortized time, using $  {O}(n^2)  $ space.
\end{proof}

Together, Theorems~\ref{thm:gensplit_gray}, \ref{thm:split_complexity}, \ref{thm:gensplitI_gray}, and \ref{thm:splitI_complexity} establish the following result.
\begin{theorem}\label{thm:split}
Spanning arborescences of split graphs can be generated in a pivot Gray code order in constant amortized time per arborescence, using $  O(n^2)  $ space.
\end{theorem}

\section{Spanning arborescences of proper interval digraphs}\label{sec:properinterval}
It is worth noting that our host–skeleton–routing framework is not restricted to split digraphs. 
In fact, it is quite versatile and can potentially be extended to broader classes of digraphs whose structure can be decomposed into connections of cliques.
We demonstrate the power of this framework by extending it to generate pivot Gray codes for spanning arborescences of proper interval digraphs rooted at the leftmost interval.

A \emph{proper interval graph} is the intersection graph of a family of intervals on the real line such that no interval properly contains another. 
The symmetric \emph{proper interval digraph} is obtained by replacing each undirected edge with a pair of symmetric directed edges. 
By Roberts' theorem~\cite{roberts1969indifference}, proper interval graphs are equivalent to \emph{unit interval graphs} (where all intervals have equal length) and \emph{indifference graphs}. 
Due to their highly regular structure, proper interval graphs play a central role in numerous applications, including chronobiological time-series analysis, task scheduling with uniform processing times, and physical mapping of DNA sequences in bioinformatics~\cite{Benzer1959, Golumbic1980}.
Proper interval graphs admit a canonical \emph{layered partition} obtained via breadth-first search (BFS). 
Given a proper interval representation, let $  r  $ be the interval with the leftmost left endpoint (which also has the leftmost right endpoint). We designate $  r  $ as the root. The vertex set can then be uniquely partitioned into BFS layers $  V_0, V_1, \dots, V_m  $ with respect to $  r  $, where $  V_0 = \{r\}  $ and, for $  i \geq 1  $,
$$V_i = \{ v \in V \mid \mathrm{dist}(r, v) = i \},$$
with $  \mathrm{dist}(r, v)  $ denoting the length of the shortest path from $  r  $ to $  v  $ in $  G  $.
These layers satisfy the following properties:
\begin{enumerate}
\item Each $  V_i  $ forms a clique;
\item There are no edges between non-consecutive layers (i.e., edges exist only within the same layer or between adjacent layers);
\item For each $  i \geq 1  $, there exists a vertex $  d_i \in V_{i-1}  $ adjacent to every vertex in $  V_i  $.
\end{enumerate}
This layering respects the natural left-to-right order of the intervals and satisfies a strong local domination property: for each $  1 \leq i \leq m  $, the vertex $  d_i \in V_{i-1}  $ dominates $  V_i  $ (i.e., $  V_i \subseteq N(d_i)  $), and the set $  \{d_i\} \cup V_i  $ induces a clique.
As an example, Figure~\ref{fig:pig_example} shows a proper interval digraph on 11 vertices, together with its interval representation and its layer partition.

\begin{figure}[t]
\centering
\includegraphics[width=.8\linewidth, trim={0 0 0 0}, clip]{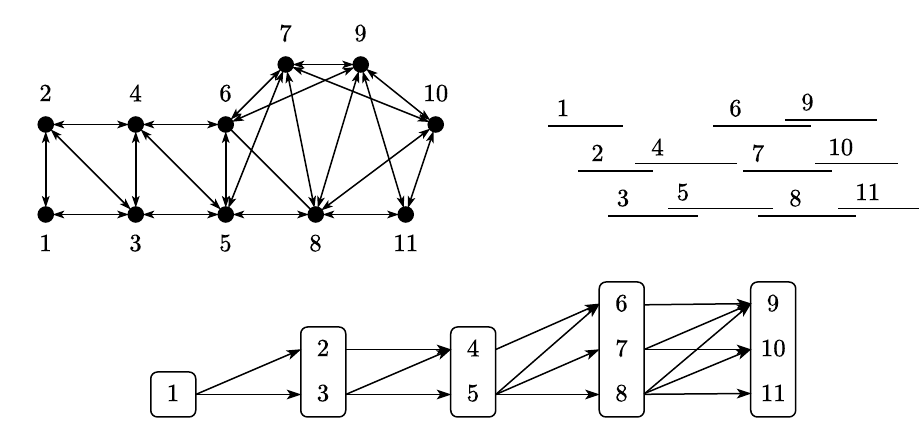}
\caption{A proper interval digraph with an interval representation and the associated layer partition.}
\label{fig:pig_example}
\end{figure}





We construct a pivot Gray code for spanning arborescences of proper interval digraphs rooted at $r \in V_0$ by recursively applying our host–skeleton–routing framework in an inductive, layer-by-layer manner.
Let $  \overrightarrow{G}_i  $ denote the induced subgraph of a proper interval digraph $  \overrightarrow{G}  $ on the vertices in layers $  V_0, V_1, \dots, V_i  $. 
Given a pivot Gray code for the spanning arborescences of $  \overrightarrow{G}_{i-1}  $, we construct the pivot Gray code for $  \overrightarrow{G}_i  $ by treating the previously solved subgraph $  \overrightarrow{G}_{i-1}  $ as the skeleton and the newly appended layer $  V_i  $ as the host. 
We show that every spanning arborescence of $  \overrightarrow{G}_i  $ rooted at $  r \in V_0  $ can be uniquely represented by the tuple $  S_i = (r, T_H, T_{i-1}, \alpha)  $.
\paragraph{1. Host ($  T_H  $):}
The \emph{host} mapping $  T_H  $ determines the parent assignments for all vertices in the newly added layer $  V_i  $. 
Unlike the independent set in a split digraph, the vertices of $  V_i  $ induce a clique. 
For any valid spanning arborescence of $  G_i  $ rooted at $  r  $, every vertex in $  V_i  $ must eventually reach the root via the predecessor layer $  V_{i-1}  $.
To achieve this, the vertices in $  V_i  $ must choose parents such that they form a spanning forest rooted at some non-empty subset $  R \subseteq V_{i-1}  $. 
Note that this structure differs slightly from the standard $  R  $-forest defined on a single vertex set, as the target roots $  R  $ lie outside $  V_i  $ and a vertex in $V_i$ is not necessrily to connect to all vertices in $V_{i-1}$. 
In contrast to split digraphs, a vertex in $  V_i  $ need not connect directly to $  V_{i-1}  $; instead, it may connect internally within the clique $  V_i  $, with only the local roots of this forest making the jump to $  V_{i-1}  $. We refer to this family of directed trees spanning $  V_i  $ as the host forest $  T_H  $.
\paragraph{2. Skeleton ($  T_{i-1}  $):}
The \emph{skeleton} $  T_{i-1}  $ is a spanning arborescence of the subproblem $  \overrightarrow{G}_{i-1}  $ (the proper interval digraph induced by layers $  V_0  $ through $  V_{i-1}  $). 
Since $  \overrightarrow{G}_{i-1}  $ is itself a proper interval digraph of smaller depth, its spanning arborescences are generated recursively by the algorithm from the previous step. The skeleton provides the global routing paths to the root $  r  $ for all vertices in the deeper layers.
\paragraph{3. Routing ($  \alpha  $):}
The \emph{routing} $  \alpha  $ determines whether a vertex $  u \in V_{i-1}  $ connects directly to its skeleton parent in $  \overrightarrow{G}_{i-1}  $, or utilizes a vertex $  v \in V_i  $ as an alias to route its path through the new layer. 
Since $  V_i  $ induces a clique with possible internal connections, this alias path may traverse multiple vertices within $  T_H  $ before returning to the skeleton.

\begin{lemma}\label{lem:pig_bijection_i}
Let $\overrightarrow{G}_i$ be the induced proper interval subgraph containing layers $V_0, V_1, \dots, V_i$ with root $r \in V_0$. There is a bijection between the spanning arborescences of $\overrightarrow{G}_i$ rooted at $r$ and the set of valid configuration tuples $S_i = (r, T_H, T_{i-1}, \alpha)$.
\end{lemma}

\begin{proof}
The proof follows that of Lemma~\ref{lem:split_bijection}, with the necessary adaptations to the layered structure of proper interval graphs. 
We replace the host mapping $  \rho  $ for split digraphs with a rooted spanning forest $  T_H  $ on $  V_i  $ whose roots lie in $  V_{i-1}  $.
Furthermore, instead of single-edge aliases as in split digraphs, each routing choice $  \alpha_u  $ (for $  u \in V_{i-1}  $) now indexes a vertex in $  V_i  $ that serves as the entry point of a directed path in $  T_H  $ terminating at the skeleton parent in $  V_{i-1}  $. 
Such a path is unique given $  T_H  $ is a forest.
The acyclicity, connectivity, and uniqueness arguments from Lemma~\ref{lem:split_bijection} carry over directly under this host–skeleton–routing definition, establishing the bijection.
\end{proof}

\begin{lemma}
Let $\overrightarrow{G}_i$ be the proper interval subgraph restricted to layers $V_0, V_1, \dots, V_i$. There is a bijection between the spanning arborescences of $\overrightarrow{G}_i$ rooted at $r \in V_0$ and the sequence of tuples $(S_1, S_2, \dots, S_i)$, where each $S_k = (r, T_{H}, T_{k-1},  \alpha_k)$ denotes the skeleton, host forest and routing aliases at layer $V_k$.
\end{lemma}
\begin{proof}
We proceed by induction on $i$. The base graph $\overrightarrow{G}_0 = (V_0 = \{r\})$ possesses only the empty arborescence. 
Assume by induction that the spanning arborescences of $\overrightarrow{G}_{i-1}$ are in bijection with the nested configurations $(S_1, S_2, \dots, S_{i-1})$. By Lemma~\ref{lem:pig_bijection_i}, the spanning arborescences of $G_i$ are in strict bijection with the tuples $(r, T_{H,i}, T_{i-1}, \alpha_i)$. Substituting the inductive assumption for the skeleton $T_{i-1}$ yields the bijection for $(S_1, S_2, \dots, S_i)$. 
\end{proof}

\begin{figure}[t]
\centering
\includegraphics[width=0.8\linewidth, trim={0 0 0 0}, clip]{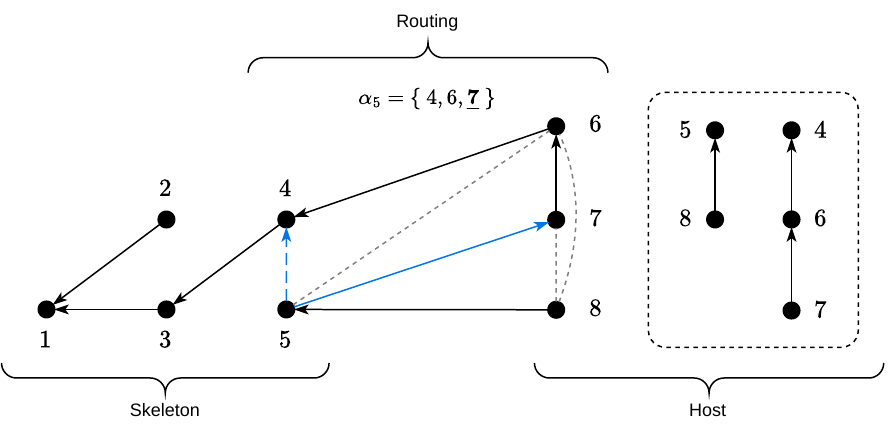}
    \caption{A spanning arborescence for the first four layers in Figure~\ref{fig:pig_example} represented by the tuple $  S_3 = (1, (101, 1), T_2, 02)  $.}
\label{fig:pig_hsr}
\end{figure}

\paragraph{Example of the tuple $S_i$ encoding (Figure \ref{fig:pig_hsr}):} 

Consider a proper interval digraph formed by the first 4 layers in Figure~\ref{fig:pig_example}. The spanning arborescence can be represented by the tuple $  S_3 = (1, (101, 1), T_2, 02)  $ where the first three layers form the skeleton $T_2$.
\begin{itemize}
    \item \textbf{Root:} $  r = 1  $.
    \item \textbf{Host ($  T_H = (101, 1)  $):} The host forms a forest in which vertices 4 and 5 in $  V_2  $ act as roots. Vertex 6 connects to the first available root in $  V_2  $ (vertex 4), and vertex 8 connects to the only available root (vertex 5). Vertex 7 then connects to the first available parent in $  V_3  $ (vertex 6).
    \item \textbf{Skeleton ($  T_{i-1} = T_2  $):} The skeleton is a spanning arborescence of $\overrightarrow{G}_2$.
    \item \textbf{Routing $ \alpha = 02 $:} Only vertex 5 in $V_2$ uses an alias path ($\alpha_5 = 7$ with its available routing options $\{4, 6, 7\}$). 
\end{itemize}

\vspace{1em}
We first prove the following lemmas for the pivot Gray code for the host generator. 

\begin{restatable}{lemma}{mixedradixconstrainted} \label{lem:mr_1gc}
    There exists a 1-Gray code for the mixed radix strings of length $m$ excluding the all-zero string $0^m$, starting from a legal string $0^ix0^{m-i-1}$, and ending at a legal string of the form $0^jy0^{m-j-1}$ where $i, j \in \{0, 1, \dots, m-1\}$, $i\neq j$, and for some $x, y > 0$.
\end{restatable}

\begin{proof}
Such a Gray code can be constructed by modifying the standard reflectable Gray code framework of Li and Sawada~\cite{li2009gray}. 
We first apply a permutation to the radices, moving those with even values to the top levels of the recursive generation tree and relegating those with odd values to the bottom. 
A detailed proof can be found in Appendix~A.
\end{proof}

\begin{lemma}
If all vertices in $  V_i  $ form a linear arborescence with one root in the first host configuration, then they also form a linear arborescence with one root in the last host configuration.
\end{lemma}
\begin{proof}
    The result follows immediately from Lemma \ref{lem:constrained_k_ary} and Lemma \ref{lem:mr_1gc}.
\end{proof}


From this we obtain the following corollary.

\begin{corollary}
Let $  T^{\mathrm{start}}_{i}  $ and $  T^{\mathrm{end}}_{i}  $ be the first and last spanning arborescences in the pivot Gray code for $  G_i  $. Then the vertices in $  V_i  $ form linear arborescences in both $  T^{\mathrm{start}}_{i}  $ and $  T^{\mathrm{end}}_{i}  $.
\end{corollary}

The pivot Gray code for the routing component can be handled similarly using the reflectable language framework. 
Note, however, that dynamic endpoint control is also required when the host generator is about to pivot a vertex that is not directly connected to $  \overrightarrow{G}_{i-1}  $, but serves as the root of an alias vertex.

\begin{lemma} \label{proper_interval_endpoint}
Let $  x \in V_i  $ be the vertex scheduled to pivot next by the host generator, and let $  T^{\mathrm{end}}_{i-1}  $ be the final arborescence of $  \overrightarrow{G}_{i-1}  $. Then the nested skeleton $  T_{i-1}  $ and routing aliases $  \alpha  $ can be dynamically controlled to terminate in a configuration where no vertex in $  \overrightarrow{G}_{i-1}  $ utilizes any vertex in the subtree rooted at $  x  $ as an alias.
\end{lemma}
\begin{proof}
The proof generalizes the endpoint control established in Lemma~\ref{thm:split_kary}. 
In a split digraph, an independent vertex $  x  $ affects only itself. 
In a proper interval digraph, the vertices of $  V_i  $ form a host forest $  T_{H}  $. 
When $  x  $ pivots its connection to $  V_{i-1}  $, it carries its entire subtree $  X \subseteq V_i  $ with it. 
If any vertex $  u \in V_{i-1}  $ utilizes a vertex $  y \in X  $ as an alias, a cycle might occur violating the arborescence property. 
Thus, we must ensure that every vertex with an alias path involving any vertex in $  X  $ ends on a direct path. 
Since the pivot Gray code for $  \overrightarrow{G}_{i-1}  $ ends with a linear arborescence, the same reset opportunities used in Lemma~\ref{thm:split_kary} can be used to force all affected vertices in $X$ to terminate on direct paths.
\end{proof}


\begin{lemma}
Given a pivot Gray code for the spanning arborescences of $  \overrightarrow{G}_{i-1}  $ and for the host layer $  V_i  $, the host–skeleton–routing framework generates a pivot Gray code for the spanning arborescences of $  G_i  $.
\end{lemma}
\begin{proof}
The proof follows that of Theorem~\ref{thm:gensplit_gray}. The host, skeleton, and routing components are independent, and each nested generator exhaustively enumerates every valid string in its domain exactly once. Thus, the overall procedure generates every valid tuple $  S_i  $ exactly once. Within each generator, the traversal is a 1-Gray code, corresponding to a single pivot change. Endpoint control, as established by Lemma~\ref{proper_interval_endpoint}, ensures that the routing configuration terminates in the direct paths where no vertex uses any vertex in the relevant subtree as an alias. 
Therefore, the listing forms a pivot Gray code.
\end{proof}

Recursively constructing $  \overrightarrow{G}_i  $ from layer $  V_0  $ to $  V_i  $ results in the following theorem.
\begin{theorem}
The host–skeleton–routing framework generates a pivot Gray code for the spanning arborescences of proper interval digraphs rooted at the leftmost interval.
\end{theorem}
Lastly, we consider the time complexity of our algorithm. 
Unlike complete graphs, there is no simple formula to compute the remaining number of arborescences on the fly for proper interval digraphs. 
Nevertheless, we achieve constant amortized time per arborescence via a two-pass simulation technique.

\begin{theorem} \label{thm:pig_complexity}
Spanning arborescences of proper interval digraphs can be generated in a pivot Gray code order in constant amortized time per arborescence, utilizing ${O}(n^2)$ space, provided the root is the leftmost interval $r \in V_0$. 
\end{theorem}

\begin{proof}
The majority of the operations follow the implementation logic developed for split digraphs. The primary difference occurs at the host block boundaries, when a host vertex in $V_i$ applies a pivot change. 

In this scenario, the algorithm must first update the routing alias mapping table. 
In split digraphs, this took ${O}(|\overrightarrow{K}|)$ time. 
In proper interval digraphs, we must verify whether the neighbors of the pivoting vertex possess valid directed paths into $V_{i-1}$. 
This path verification and alias update requires at most ${O}(|V_i| \cdot |V_{i-1}|)$ time. 
However, this setup cost is incurred only at the boundary of a host block. Immediately following this update, the algorithm traverses the inner domains of the skeleton and routing configurations. 
Since this inner generation produces a number of arborescences that strictly dominates the polynomial setup cost (or is bounded by a trivial constant in base cases), the host update mathematically amortizes to constant time per generated arborescence.

Then, the dynamic endpoint control needs to determine the parity of the remaining sequence. 
Rather than relying on heavy algebraic precomputation, we achieve this via a two-pass simulation technique. 
For a fixed host configuration, the algorithm first executes a ``simulation'' pass of the skeleton generator. 
During this initial pass, the algorithm does not evaluate or yield the inner routing aliases; it merely traverses the skeleton configurations and records the exact rank (step index) at which each clique vertex updates its skeleton parent. 
Since the skeleton generator inherently operates in constant amortized time, this simulation run requires time proportional to the number of skeleton arborescences and should amortize to constant time. 
The spatial overhead requires only a simple ${O}(n)$ array to store the terminal ranks. During the second pass, the algorithm compares the current step index against the pre-recorded rank to detect the final transition in constant time to calculate the exact parity of the remaining sequence to orient the endpoint control. 
Each layer uses the same optimization as {\sc GenSplitArb} to maintain a local routing table between $V_{i-1}$ and $V_i$, utilizing $O(n_i^2)$ space locally and $\sum_\ell O(n_i^2) \le O(n^2)$ globally.
Therefore, the entire framework generates spanning arborescences for proper interval digraphs in constant amortized time, utilizing ${O}(n^2)$ space.
\end{proof}

\section{Final remarks}\label{sec:conclusion}
The host-skeleton-routing framework introduced in this paper opens several interesting avenues for future research. 
A natural extension is generalizing this framework to other graph classes characterized by overlapping clique intersections. 
It is also worth noting that our proper interval digraph algorithm fundamentally relies on the BFS layer partition, which requires the arborescences to be rooted at the leftmost interval $r \in V_0$ to guarantee the unidirectional domination property. 
If the desired root lies deep within the interval sequence (e.g., $r \in V_k$ for $k > 0$), the directed paths might need to propagate both leftward and rightward simultaneously, breaking the strict layer-by-layer structure. 
Whether an efficient algorithm exists for generating pivot Gray codes for spanning arborescences of proper interval digraphs rooted at an arbitrary vertex remains an open problem.

\section*{AI disclosure} 
We used ChatGPT and Google Gemini to assist with the implementation, debugging, and polishing of the Python code presented in the appendices. The tool materially contributed to the development and experimental validation of the sampling algorithms. The authors verified the correctness and originality of all content including references.


\label{sec:remark}


\newpage
\bibliographystyle{abbrv}
\bibliography{myrefs}       

\newpage
\noindent
{\large
{\bf Appendix A: Proofs of Lemmas related to reflectable languages}}

\hfill



\constrainedkary*
\begin{proof}
When $  k=1  $ (binary case), the statement follows directly from BRGC as shown in Lemma~\ref{lem:constrained_k_ary}. 
For $  k > 1  $, we configure each level of the reflectable recursive tree so that it begins with the desired starting digit. 
For every level, we then designate the leftmost and rightmost values of each local reflectable block as $  k-1  $ and $  k-2  $, respectively.
This effectively hides the digit 0 during all intermediate reflections and only exposes it at the very last block of the generation tree. 
This configuration guarantees that the entire sequence terminates exactly on the all-zero string $0^m$. 
By simply dropping this final string, the sequence ends on its immediate predecessor, which differs from $0^m$ in exactly one position and naturally has the desired form $0^j y 0^{m-j-1}$.
Figure~\ref{fig:placeholder} illustrates the structure of the first and second levels of the reflectable recursion tree when $  k=2  $.
\end{proof}

\mixedradixconstrainted*
\begin{proof}
Such a Gray code can be constructed by modifying the standard reflectable Gray code framework of Li and Sawada~\cite{li2009gray}. We first apply a permutation to the radices, moving those with even values to the top levels of the recursive generation tree and relegating those with odd values to the bottom. Let the reordered radices be $  r_1, r_2, \dots, r_m  $.
We configure the reflectable sequence based on the parity of the available radices.

\medskip
\noindent\textbf{Case 1: All radices are odd.} \
In this case, $  r_k \ge 3  $ for all $  k  $. We set the root level $  r_1  $ to start at 1 and end at 0. For all subsequent levels, we designate the leftmost and rightmost values of each local reflectable block as 1 and 2, respectively. This effectively hides the digit 0 during all intermediate reflections and only exposes it at the very last block of the generation tree. This configuration guarantees that the entire sequence terminates exactly on the all-zero string $0^m$. By simply dropping this final string, the sequence ends on its immediate predecessor, which differs from $0^m$ in exactly one position and naturally has the desired form $0^j y 0^{m-j-1}$.

\medskip
\noindent\textbf{Case 2: At least one radix is even.} \
Here the root level $  r_1  $ is even. We configure $  r_1  $ to start at 1 and end at 0. 
Since $  r_1  $ is even, the recursive tree undergoes an even number of reflections, which preserves endpoint parity. 
For each subsequent even-radix level, the number of strings generated at the parent level is even, so the recursive subtree again undergoes an even number of reflections, preserving endpoint parity.
This allows us to safely start and end all subsequent even-radix levels on 0. When the recursion reaches the odd radices at the bottom, we apply the 0-hiding technique from Case 1. 
The sequence terminates at $0^m$ which is then dropped.

\medskip
\noindent\textbf{Digit $  x  $ at an odd-radix level:}
To force the sequence to start with a string in which the unique non-zero digit is located at an odd radix position, we map the target radix to the last level of the recursion tree. We configure this leaf level to start at $1$ and transition to $0$, while setting the first level to start at $0$ and end at $1$. This forces the all-zero string $0^m$ to appear as the second string in the overall sequence, sandwiched between $0^{m-1}1$ and $0^{m-1}t$ for some $  t \ge 2  $. The last string then becomes $10^{m-2}$. Since the predecessor and successor of $0^m$ differ only in this single leaf coordinate, they are adjacent in the 1-Gray code graph. Removing $0^m$ therefore directly connects $0^{m-1}1$ to $0^{m-1}t$, achieving the desired endpoint.
\end{proof}

\begin{figure}[t]
    \centering
    \includegraphics[width=0.7\linewidth]{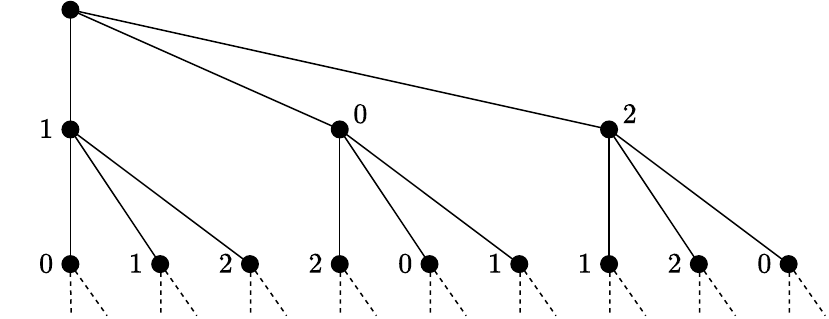}
    \caption{The structure of the first and second levels of the reflectable recursion tree when $  k=2  $.}
    \label{fig:placeholder}
\end{figure}

\newpage
\noindent
\large
{\bf Appendix B: Pivot Gray code of $\overrightarrow{K}_5$ generated by {\sc GenArb}}

\tiny
\begin{figure}[H]
    \centering
    \includegraphics[width=0.99\linewidth, trim={5 20 0 5}, clip]{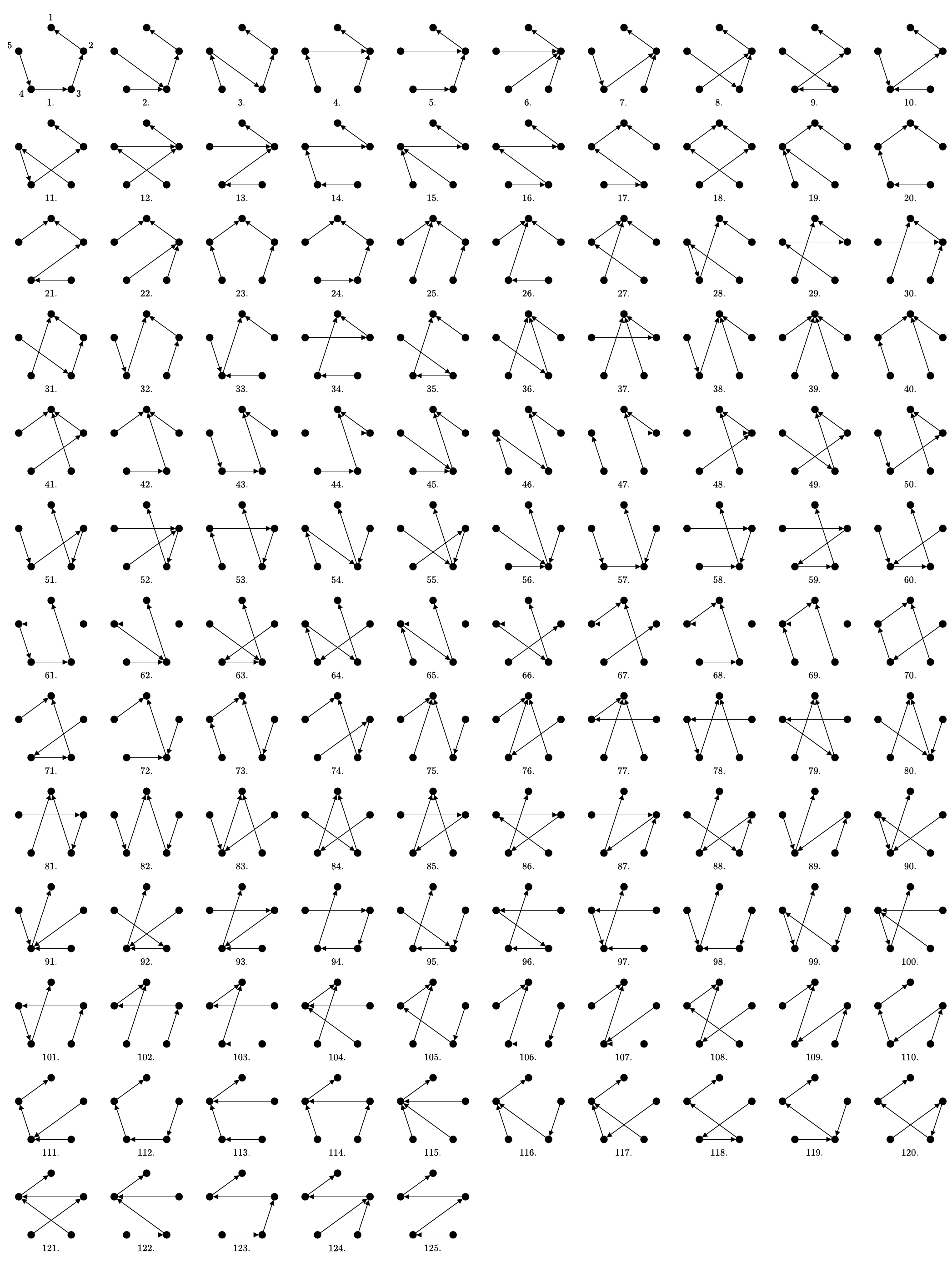}
\end{figure}

\newpage
\noindent
{\large
{\bf Appendix C: Python code to generate pivot Gray code for spanning arborescences of complete digraphs}}

\lstset{style=mystyle}
\begin{lstlisting}[language=Python]
from k_ary_linear import k_ary_linear


class Node:
    def __init__(self, val):
        self.val = val
        self.parent = None
        self.leaf = self  # bottom leaf
        self.nchild = 0


def set_parent(node, parent):
    old = node.parent
    if old is not None:
        old.nchild -= 1
        if old.nchild == 0:
            old.leaf = old
    node.parent = parent
    if parent is not None:
        parent.nchild += 1


node_list = []


def spanning_level(V, t_, n_c):
    if n_c == 0:
        yield node_list
        return

    P = [V[i] for i in range(len(V)) if t_[i] > 0]      
    C = [V[i] for i in range(len(V)) if t_[i] == 0]      
    t_l = [0 for _ in range(n_c)]
    roots = {}                                           

    m = {P[i].val: i + 1 for i in range(len(P))}
    for i in range(len(C)):
        if C[i].parent.val in m:
            t_l[i] = m[C[i].parent.val]
            roots[i] = True

    yield from spanning_level(C, t_l, len(C) - len(roots))

    K = k_ary_linear(len(P), t_l)
    next(K)
    for t_l, idx in K:
        if t_l[idx[0]] > 0:                              
            set_parent(C[idx[0]], P[t_l[idx[0]] - 1])
            if idx[0] not in roots:
                roots[idx[0]] = True
        else:                                            
            del roots[idx[0]]
            last = next(reversed(roots)) if roots else -1
            set_parent(C[idx[0]], C[last].leaf)
        yield from spanning_level(C, t_l, len(C) - len(roots))

    for c_idx in roots:                                  
        C[c_idx].parent.leaf = C[c_idx].leaf


if __name__ == '__main__':
    def str_tree(nodes):
        return ''.join('{}->{}; '.format(nd.val + 1, nd.parent.val + 1 if nd.parent else -1)
                       for nd in nodes)

    n = int(input('Enter n: '))
    prev = None
    for i in range(n):
        nd = Node(i)
        if prev is not None:
            set_parent(nd, prev)
        node_list.append(nd)
        prev = nd
    count = 0
    for tree in spanning_level(node_list, [1] + [0] * (n - 1), n - 1):
        print(str_tree(tree))
        count += 1
    print('Total:', count)

\end{lstlisting}

\newpage

\noindent
\large
{\bf Appendix D: Python code to generate 1-Gray code $k$-ary strings and mixed-radix strings}\label{app:k-ary}

\tiny
\lstset{style=mystyle}
\begin{lstlisting}[language=Python]
def k_ary_linear(k, alpha):

    n = len(alpha)
    d = [1 if a != 1 else -1 for a in alpha]
    idx = []
    s = [sum(alpha)]


    order = [i for i in range(n)]

    if k == 1:        
        for i in range(n):
            if alpha[i] == 1:
                order[0], order[i] = i, 0
                break
    
    def _next(i, last=-1):
        i_ = order[i]

        for r in range(k + 1):

            nl = 0 if k > 1 and r == k and last != -1 else -1

            if i == n-1:
                if s[0] != 0:
                    if len(idx) > 1 and idx[0] == idx[1]:
                        idx.pop()
                    yield alpha, idx
                    idx[:] = []
            else:
                i_next = order[i + 1]
                yield from _next(i + 1, nl)
                d[i_next] = 1 if alpha[i_next] == k else -d[i_next]

            if r < k:
                a_b = alpha[i_]
                a_a = (alpha[i_] + d[i_] + k + 1) % (k + 1)
                if last != -1:
                        if r == k-1:
                            a_a = last
                        elif a_a == last:
                            a_a = (a_a + d[i_] + k + 1) % (k + 1)
                s[0] += a_a - a_b
                alpha[i_] = a_a
                idx.append(i_)

    yield from _next(0, -1 if k == 1 else (k if alpha[order[0]] != k else k-1))
    



\end{lstlisting}

\newpage
\noindent
{\large
{\bf Appendix E: Python code to generate pivot Gray code for spanning arborescences of split digraphs}}

\lstset{style=mystyle}
\begin{lstlisting}[language=Python]
from k_ary_linear import k_ary_linear

class Node:
    def __init__(self, val):
        self.val = val
        self.parent = None
        self.leaf = self
        self.nchild = 0   

class LinkedList:
    def __init__(self):
        self.map = {}
        self.first = -1
        self.last = -1
        self.len = 0

    def append(self, val):
        assert val not in self.map
        self.map[val] = [self.last, -1, None]
        if self.last != -1:
            self.map[self.last][1] = val
        self.last = val
        self.len += 1
        if self.first == -1:
            self.first = val

    def remove(self, val):
        assert val in self.map
        val_last = self.map[val][0]
        val_next = self.map[val][1]
        if val_last != -1:
            self.map[val_last][1] = val_next
        if val_next != -1:
            self.map[val_next][0] = val_last
        else:
            self.last = val_last
        if self.first == val:
            self.first = val_next
        del self.map[val]
        self.len -= 1
    
    def next(self, val):
        if val not in self.map:
            return -1
        if self.map[val][1] != -1:
            return self.map[val][1]
        else:
            return self.first
    
    def prev(self, val):
        if val not in self.map:
            return -1
        if self.map[val][0] != -1:
            return self.map[val][0]
        else:
            return self.last
    
    def insert_first(self, val):
        assert val not in self.map
        self.map[val] = [-1, self.first, None]
        if self.first != -1:
            self.map[self.first][0] = val
        self.first = val
        self.len += 1
        if self.last == -1:
            self.last = val

N_K = [-1]
N_I = [-1]

k_set_list = []     # k vertices list, [i] -> Node object
k_cic_list = []     # cross I children list, [i][j] -> mixed radix if i is parent of j
i_set_list = {}     # i vertices list, [i] -> Node object
i_edge = {}


host_list = {}      # feasible host of i vertices, [i] -> LinkedList of feasible host k vertices
host_states = {}    # current host of i vertices, [i] -> current host k vertex
leaf_i_vertex = [-1]  # i vertex that must be a leaf when H1 -> H2

routing_list = LinkedList()     # routing list of k.val, [k.val] -> [last, next, parent]
routing_states = {}             # current routing val of each k

root_node_ref = [None]   # Node of the I-vertex acting as root
root_ports = [None]      # list of adjacent K-vertices of the I-root

def set_parent(node, parent):
    old = node.parent
    if old is not None:
        old.nchild -= 1
        if old.nchild == 0:
            old.leaf = old
    node.parent = parent
    if parent is not None:
        parent.nchild += 1

def host():
    i_keys = list(i_edge.keys())    

    def host_mr(i_idx):
        if i_idx == N_I[0]:
            return
        
        i_val = i_keys[i_idx]
        yield from host_mr(i_idx + 1)
        is_first = True if host_states[i_val] == host_list[i_val].first else False  # reflectable
        for _ in range(host_list[i_val].len - 1):  
            leaf_i_vertex[0] = i_val
            yield  # return next change before doing this

            p_current = host_states[i_val]
            p_next = host_list[i_val].next(p_current) if is_first else host_list[i_val].prev(p_current) 
            host_states[i_val] = p_next
            i_set_list[i_val].parent = k_set_list[p_next]

            # update feasible paths
            for p in i_edge[i_val]:
                if p != p_current:
                    k_cic_list[p_current][p].remove(i_val)
                if p != p_next:
                    k_cic_list[p_next][p].append(i_val)

            # update feasible routings
            for k in k_set_list:
                if k.parent is None or k.parent.val >= N_K[0]:
                    continue
                if k_cic_list[k.parent.val][k.val].len >= 2:
                    if k.val not in routing_list.map:
                        routing_list.append(k.val)
                        routing_states[k.val] = k.parent.val
                    routing_list.map[k.val][2] = k.parent.val
                else:
                    if k.val in routing_list.map:
                        routing_list.remove(k.val)
                    routing_states[k.val] = -1
            
            yield from host_mr(i_idx + 1)

    yield from host_mr(0)
    leaf_i_vertex[0] = -1
    yield -1

def spanning_level(V, t_, n_c, last=False, first_m=None):
    if n_c == 0:
        yield 
        return

    P = [V[i] for i in range(len(V)) if t_[i] > 0]
    C = [V[i] for i in range(len(V)) if t_[i] == 0]
    t_l = [0 for _ in range(n_c)]
    C_ = LinkedList()

    width = len(C) if first_m is None else first_m

    n_str = (len(P) + 1) ** width - 1

    m = {}
    for i in range(len(P)):
        m[P[i].val] = i+1

    for i in range(len(C)):
        if C[i].parent.val in m:
            t_l[i] = m[C[i].parent.val]
            C_.append(i)

    yield from spanning_level(C, t_l, len(C) - C_.len, last and n_str == 1)

    n_str -= 1
    tail_holder = [0] * (len(C) - width)

    K = k_ary_linear(len(P), t_l[:width])
    next(K)

    for ap, idx in K:
        t_l = ap + tail_holder   
        this_c = C[idx[0]].val
        
        # skeleton update
        if t_l[idx[0]] > 0:  
            set_parent(C[idx[0]], P[t_l[idx[0]]-1])
            if not idx[0] in C_.map:
                C_.append(idx[0])
            this_p = P[t_l[idx[0]]-1].val
        else:
            C_.remove(idx[0])
            set_parent(C[idx[0]], C[C_.last].leaf)
            this_p = C[C_.last].leaf.val
        
        # routing table update
        if this_p >= N_K[0]:
            # parent is the I-root: the (child -> root) edge is never routable
            if this_c in routing_list.map:
                routing_list.remove(this_c)
            routing_states[this_c] = -1
        elif k_cic_list[this_p][this_c].len == 1:
            if this_c in routing_list.map:
                routing_list.remove(this_c)
                routing_states[this_c] = -1
        else:

            # when 0..011 available
            if (
                last 
                and n_str == 2 
                and k_cic_list[this_p][this_c].len >= 2 
                and leaf_i_vertex[0] in k_cic_list[this_p][this_c].map
            ):
                if this_c in routing_list.map:
                    routing_list.remove(this_c)
                routing_list.insert_first(this_c)
                routing_list.map[this_c][2] = this_p
                routing_states[this_c] = this_p if n_c % 2 == 0 else k_cic_list[this_p][this_c].last

            # when 11 -> 01
            elif(
                last
                and n_str == 1
                and len(t_l) == 2
                and leaf_i_vertex[0] in k_cic_list[this_p][this_c].map
            ):
                if this_c in routing_list.map:
                    routing_list.remove(this_c)
                routing_list.insert_first(this_c)
                routing_list.map[this_c][2] = this_p
                routing_states[this_c] = leaf_i_vertex[0]
            
            # general case
            else:
                if this_c not in routing_list.map:
                    routing_list.append(this_c)
                routing_list.map[this_c][2] = this_p
                routing_states[this_c] = this_p

        yield from spanning_level(C, t_l, len(C) - C_.len, last and n_str == 1)
        n_str -= 1

    # update leaf
    for c_idx in C_.map:
        C[c_idx].parent.leaf = C[c_idx].leaf

def routing():

    def routing_mr(i):
        if i == -1:
            yield
            return
        
        mr_list = k_cic_list[routing_list.map[i][2]][i]
        yield from routing_mr(routing_list.map[i][1])

        is_first = True if mr_list.first == routing_states[i] else False  # reflectable

        for _ in range(mr_list.len-1):
            routing_states[i] = mr_list.next(routing_states[i]) if is_first else mr_list.prev(routing_states[i])
            yield from routing_mr(routing_list.map[i][1])

    yield from routing_mr(routing_list.first)

def gen_split():
    for H in host():
        for S in spanning_level(
            k_set_list,
            [1] + [0 for _ in range(len(k_set_list) - 1)],
            len(k_set_list) - 1,
            leaf_i_vertex[0] != -1
            ):
            for R in routing():
                yield

def gen_split_iroot():
    root_node = root_node_ref[0]
    ports = root_ports[0]
    port_set = set(ports)
    V = ([root_node]
         + [k_set_list[p] for p in ports]
         + [k_set_list[j] for j in range(N_K[0]) if j not in port_set])
    m = len(ports)
    for H in host():
        for S in spanning_level(
            V,
            [1] + [0 for _ in range(N_K[0])],
            N_K[0],
            leaf_i_vertex[0] != -1,
            first_m=m,
            ):
            for R in routing():
                yield

def initialize_split_graph(n, edges):

    prev = None
    for v in range(n):
        new_node = Node(v)
        if prev is not None:
            set_parent(new_node, prev)
        k_set_list.append(new_node)
        k_cic_list.append([LinkedList() for _ in range(n)])
        for cic in k_cic_list[-1]:
            cic.append(v) 
        prev = new_node

    for i_val, i_parents in edges.items():
        assert i_val not in i_set_list
        assert len(i_parents) > 0
        assert len(set(i_parents)) == len(i_parents)
        assert all(0 <= p < n for p in i_parents)

        i_edge[i_val] = list(i_parents)

        new_node = Node(i_val)
        i_set_list[i_val] = new_node
        new_node.parent = k_set_list[i_parents[0]]

        i_ll = LinkedList()
        for p in i_parents:
            i_ll.append(p)
        host_list[i_val] = i_ll
        host_states[i_val] = i_parents[0]

        for p in i_parents[1:]:
            k_cic_list[i_parents[0]][p].append(i_val)

    N_K[0] = n
    N_I[0] = len(i_edge)

    for k in k_set_list[1:]:
        if k_cic_list[k.parent.val][k.val].len > 1:
            routing_list.append(k.val)
            routing_list.map[k.val][2] = k.parent.val
            routing_states[k.val] = k.parent.val

def initialize_split_graph_iroot(n, edges, root):

    assert root in edges
    assert len(edges[root]) > 0
    assert len(set(edges[root])) == len(edges[root])
    assert all(0 <= p < n for p in edges[root])

    ports = list(edges[root])
    port_set = set(ports)

    for v in range(n):
        new_node = Node(v)
        k_set_list.append(new_node)
        k_cic_list.append([LinkedList() for _ in range(n)])
        for cic in k_cic_list[-1]:
            cic.append(v)

    # the I-root
    root_node = Node(root)
    root_node_ref[0] = root_node
    root_ports[0] = ports

    C_order = [k_set_list[p] for p in ports] + [k_set_list[j] for j in range(n) if j not in port_set]
    set_parent(C_order[0], root_node)
    for i in range(1, len(C_order)):
        set_parent(C_order[i], C_order[i - 1])

    for i_val, i_parents in edges.items():
        if i_val == root:
            continue
        assert i_val not in i_set_list
        assert len(i_parents) > 0
        assert len(set(i_parents)) == len(i_parents)
        assert all(0 <= p < n for p in i_parents)

        i_edge[i_val] = list(i_parents)

        new_node = Node(i_val)
        i_set_list[i_val] = new_node
        new_node.parent = k_set_list[i_parents[0]]

        i_ll = LinkedList()
        for p in i_parents:
            i_ll.append(p)
        host_list[i_val] = i_ll
        host_states[i_val] = i_parents[0]

        for p in i_parents[1:]:
            k_cic_list[i_parents[0]][p].append(i_val)

    N_K[0] = n
    N_I[0] = len(i_edge)     # excludes the root

    # initial routable clique edges
    for k in k_set_list:
        if k.parent is None or k.parent.val >= n:
            continue
        if k_cic_list[k.parent.val][k.val].len > 1:
            routing_list.append(k.val)
            routing_list.map[k.val][2] = k.parent.val
            routing_states[k.val] = k.parent.val

def decode_tree():

    tree = {
        k.val: (-1 if k.parent is None else k.parent.val)
        for k in k_set_list
    }
    tree.update({i.val: i.parent.val for i in i_set_list.values()})
    if root_node_ref[0] is not None:
        tree[root_node_ref[0].val] = -1

    for c, parent in routing_states.items():
        if parent != -1:
            tree[c] = parent

    return tree

def str_hsr_tree():
    tree = decode_tree()
    return ''.join(
        '{}->{}; '.format(v, tree[v] if tree[v] != -1 else -1)
        for v in sorted(tree)
    )


if __name__ == '__main__':
    initialize_split_graph_iroot(4, {4: [0, 1, 2], 5: [1, 3]}, root=4)
    
    print('==============================================')
    cnt = 0
    for _ in gen_split_iroot():
        print(str_hsr_tree())
        cnt += 1
    print('Total: ' + str(cnt))

    # initialize_split_graph(4, {4: [0, 1, 2], 5: [1, 3]})
    # print('==============================================')
    # cnt = 0
    # for _ in gen_split():
    #     print(str_hsr_tree())
    #     cnt += 1
    # print('Total: ' + str(cnt))

\end{lstlisting}

\end{document}